\documentclass[a4paper,10pt,allowtoday, allowfontchangeintitle, accepted=2025-03-07]{quantumarticle}
\pdfoutput=1
\usepackage[sort&compress,numbers]{natbib}
\usepackage{amssymb}

\usepackage{amsthm}
\usepackage{mathtools}
\usepackage{float}
\usepackage{sidecap}
\usepackage{color}
\usepackage{graphicx}
\usepackage{dcolumn}
\usepackage{bm}

\definecolor{niceblue}{rgb}{.1, .25, .8}
\usepackage[breaklinks=true,
hyperfootnotes=false,
            colorlinks=true,
            linkcolor=niceblue,
            urlcolor=blue,
            citecolor=niceblue
            ]{hyperref}
\hypersetup{linktoc=all}

\usepackage[left=2cm,right=2cm,top=2cm,bottom=2cm ,bindingoffset=0cm]{geometry}

\usepackage{natbib}
\usepackage{graphicx}
\usepackage{color}
\usepackage{braket}

\setcitestyle{square}

\def\Tr{\mathrm{Tr}}
\def\det{\mathrm{det}}

\newtheorem{theorem}{Theorem}

\newcommand{\myleft}{\mathopen{}\mathclose\bgroup\left}
\newcommand{\myright}{\aftergroup\egroup\right}

\newcounter{example}[section]

\allowdisplaybreaks[4]

\newcommand{\defeq}{\equiv}

\newcommand{\bs}{\backslash}

\newcommand{\GG}{\mathcal{G}}

\begin{document}

\title{Classifying fermionic states via many-body correlation measures}
\author{Mykola Semenyakin}
\affiliation{Perimeter Institute for Theoretical Physics, Waterloo, ON N2L 2Y5, Canada}
\affiliation{Instituut-Lorentz, Universiteit Leiden, P.O. Box 9506, 2300 RA Leiden, The Netherlands}
\author{Yevheniia Cheipesh}
\affiliation{Instituut-Lorentz, Universiteit Leiden, P.O. Box 9506, 2300 RA Leiden, The Netherlands}
\author{Yaroslav Herasymenko}
\email{yaroslav@cwi.nl}
\affiliation{QuSoft and CWI, Science Park 123, 1098 XG Amsterdam, The Netherlands}
\affiliation{QuTech, TU Delft, P.O. Box 5046, 2600 GA Delft, The Netherlands}
\affiliation{Delft Institute of Applied Mathematics, TU Delft, 2628 CD Delft, The Netherlands}

\begin{abstract}
Understanding the structure of quantum correlations in a many-body system is key to its computational treatment. For fermionic systems, correlations can be defined as deviations from Slater determinant states. The link between fermionic correlations and efficient computational physics methods is actively studied but remains ambiguous. 
We make progress in establishing this connection mathematically. In particular, we find a rigorous classification of states relative to $k$-fermion correlations, which admits a computational physics interpretation. 
Correlations are captured by a measure $\omega_k$, a function of $k$-fermion reduced density matrix that we call twisted purity.  
A condition $\omega_k=0$ for a given $k$ puts the state in a class $\GG_k$ of correlated states.
Sets $\GG_k$ are nested in $k$, and Slater determinants correspond to $k = 1$.
Classes $\GG_{k={\rm O}(1)}$ are shown to be physically relevant, as $\omega_k$ vanishes or nearly vanishes for truncated configuration-interaction states, perturbation series around Slater determinants, and some nonperturbative eigenstates of the 1D Hubbard model.  
For each $k = {\rm O}(1)$, we give an explicit ansatz with a polynomial number of parameters that covers all states in $\GG_k$. Potential applications of this ansatz and its connections to the coupled-cluster wavefunction are discussed.
\end{abstract}

\maketitle

\textit{Introduction.} 
Quantum correlations are central to many-body quantum problems, their computational treatment, and complexity. For fermionic systems with a fixed particle number, a natural definition of uncorrelated states are Slater determinants \cite{ecke02}. These states arise in noninteracting systems and admit efficient computations. 
To characterize the correlations, or deviations of a state from a Slater determinant, quantitative measures are employed.
Some measures were proposed in terms of one- and two-fermion reduced density matrices \cite{zies95, luzan07, grob94, davidson12, bravyi17}, minimal distance to the manifold of Slater states \cite{vedr98, gott05} and so-called Slater rank \cite{ecke02}. 
Fermionic magic \cite{cudby23, reardon23, dias23}, measuring the deviation of a quantum circuit from fermionic linear optics \cite{bravyi06, heben19, knill01, terhal02, bravyi05}, can also be considered a quantifier of correlations.
Correlation measures were applied to characterize physical systems, as well as to guide computational physics and chemistry methods \cite{zana02, zies95, bogu12, gu03, ding20, lege03,lee89, niel99, stair20, ferreira22}. 
Such methods are often based on approximate ansatzes, such as configuration-interaction states \cite{crem13, hofm03}, coupled-cluster ansatzes \cite{coes60,cizek66,kumm78,bish91,bart07}, Jastrow and Gutzwiller wavefunctions \cite{jast55, krot77,hell91}, tensor networks \cite{chan11, krum16, corb10, cira21}, and generalized Gaussian ansatzes \cite{bravyi17, shi18, hackl20}. It is a widespread heuristic, that a state with bounded correlations should admit a representation by a compact ansatz \cite{roth09,chan11, leht17, hait19}. The mathematical understanding of this connection, however, remains limited.

In this work, we give a classification of $k$-fermion correlations with a direct link to computational physics. The key new object is the twisted purity $\omega_k$, a correlation measure which is a function of the $k$-body reduced density matrix, invariant under single-particle rotations. Twisted purity is also a Hermitian observable on two copies of a state. If a state yields $\omega_k=0$, its amplitudes obey a generalization of so-called Plücker relations \cite{griff94}. We denote the set of such states $\GG_k$. Condition $\omega_{k_1}=0$ implies $\omega_{k_2}=0$ for all $k_2>k_1$, so the sets $\GG_k$ are nested in $k$ and define a classification of fermionic states ($\GG_1$ are Slater). We prove that each $\GG_{k=O(1)}$ is covered by a poly-sized ansatz. This ansatz is similar to the coupled-cluster wavefunction but has a different functional form. We partially uncover the physical meaning of classes $\GG_k$, finding that some (approximate) 1D Hubbard model eigenstates, general perturbative series around Slater determinants, and truncated configuration-interaction states belong to $\GG_k$. By implication, twisted purity diagnoses the reducibility of a state to configuration-interaction form.

Our results complement the categorization of states via entanglement scaling \cite{hastings07, li08, eise10, cira21}. On the one hand, akin to twisted purity, entanglement can be viewed as a correlation measure, and its classification underlies computational physics methods that employ tensor networks. On the other hand, a major qualitative difference is that we focus on correlations insensitive to single-particle basis rotations, while these, in general, produce volume-law entanglement \cite{bian22}. Furthermore, compared to the highly developed studies of bipartite entanglement, the type of correlations encoded in the $k$-fermion reduced density matrix is much less understood.

This paper is organized as follows. First, we introduce some notations and basic concepts used in the manuscript. We then define twisted purity $\omega_k$ and the classes $\GG_k$. Next, we examine various states in $\GG_k$ and analyze the meaning of these classes. In the last part of the text, we show that states in $\GG_k$ admit a generalized Wick's rule for the amplitudes. It implies the advertised representation by a compact non-Slater ansatz.

\textit{Notations.}
We focus on the fermionic Hilbert space $\mathcal{H}$ on $l$ modes with $n$ particles. 
Using a set of creation operators $\{\psi^\dag_r\,|\, r\in [l]\}$ and the Fock vacuum $\ket{\varnothing }$, we define in $\mathcal{H}$ a basis of states $\ket{S}=\Psi^{\dag}_{S}\ket{\varnothing }=\psi^{\dag}_{s_n}..\psi^{\dag}_{s_1}\ket{\varnothing }$.  
Here $S$ is an \textit{ordered} integer sequence $(s_1,..,s_n)\subset [l]$; unless specified otherwise, we use capital Latin letters for such sequences. Any state $\ket{v}\in\mathcal{H}$ can be decomposed into amplitudes $v(S)$ as $\ket{v}=\sum v(S)\ket{S}$. 
For convenience, we treat ordered sequences as sets and use operations such as $S_1\bs S_2$ and $S_1\cup S_3$ for $S_{1,2,3} \subset [l]$. Unions $\cup$ are always disjoint. We denote with $\sigma(S_1, S_2)$ the sign of the permutation which sorts a concatenation of $S_{1}$ and $S_2$. For more details on fermionic algebra and sequence manipulations, the reader may refer to Appendix\,\ref{app:ferm_seq}.

\textit{Plücker relations.} A general Slater determinant (or free-fermion) state $\ket{v}$ is given as
\begin{align}
\ket{v}=U\ket{S_0},~U=\exp\left(i \textstyle\sum (\theta_{pq} \psi^\dag_p \psi_q +\theta_{pq}^*\psi^\dag_q \psi_p )\right),
\label{eq:Slater_ob_rot}
\end{align}
for some fixed reference $\ket{S_0}=\Psi^{\dag}_{S_0}\ket{\varnothing }$ and complex numbers $\theta_{pq}$. Unitary $U$ in Eq.\,\ref{eq:Slater_ob_rot} is referred to as a single-particle transformation --- since it can be used to change the basis of single-particle fermionic modes $\{\psi_r\,|\, r\in[l]\}$.
Defining conditions for state $\ket{v}$ to be Slater can be phrased using the operator \cite{smirnov13}
\begin{align}
\Omega\defeq \textstyle\sum^l_{r=1} \psi_r \otimes \psi^\dag_r\label{eq:Plücker_op},
\end{align}
which acts on $\mathcal{H}\otimes \mathcal{H}$. These conditions, dubbed Plücker relations, are the components of equation \cite{miwa00}
\begin{align}
\Omega \ket{v}\otimes \ket{v}=0.
\label{eq:Plücker_rel}
\end{align}
Although Eqs.\,\ref{eq:Plücker_op}-\ref{eq:Plücker_rel} are given in a particular basis of fermionic modes $\{\psi_r \}$, $\Omega$ does not depend on such a basis (is \textit{single-particle invariant}). Namely, $[\Omega, U\otimes U]=0$ for any single-particle transformation $U$.
This fact is at the core of Plücker relations being a necessary and sufficient indicator of a Slater state.

The Plücker relations are usually formulated as algebraic equations in amplitudes \cite{griff94, levay05, levay08}, and the phrasing in terms of $\Omega$ \cite{miwa00, smirnov13} is less traditional. Let us demonstrate their equivalence for the simplest case, $(l,n)=(4,2)$. Projecting Eq.\,\ref{eq:Plücker_rel} onto $\bra{1}\otimes\bra{2,3,4}$ yields 
\begin{align}
v(1,2)v(3,4)-v(1,3)v(2,4)+v(1,4)v(2,3)=0,
\label{eq:4_2_Plücker}
\end{align}
which in this case is the only independent relation coming from Equation \ref{eq:Plücker_rel}. Equation \ref{eq:4_2_Plücker} is indeed the standard Plücker relation for $2$ fermions on $4$ modes \cite{levay05, schliemann01}.

Plücker relations can be connected to a scalar correlation measure. 
Equation \,\ref{eq:Plücker_rel} amounts to the vanishing of purity $\omega_1$ of the one-body reduced density matrix $\rho_1$,
\begin{align}
    0&=|\Omega \ket{v}\otimes \ket{v}|^2 =\rm{Tr}(\rho_1-\rho_1^2)\equiv\omega_1, \label{eq:scalar_Plücker}\\
    \rho_1^{pq}&\equiv \bra{v}\psi^\dag_q\psi_p\ket{v},
\end{align}
which is another known criterion for $\ket{v}$ to be Slater \cite{zies95, plastino09}. We will momentarily generalize these ideas to correlated (non-Slater) states.

\begin{figure*}[t]
\includegraphics[width=\linewidth]{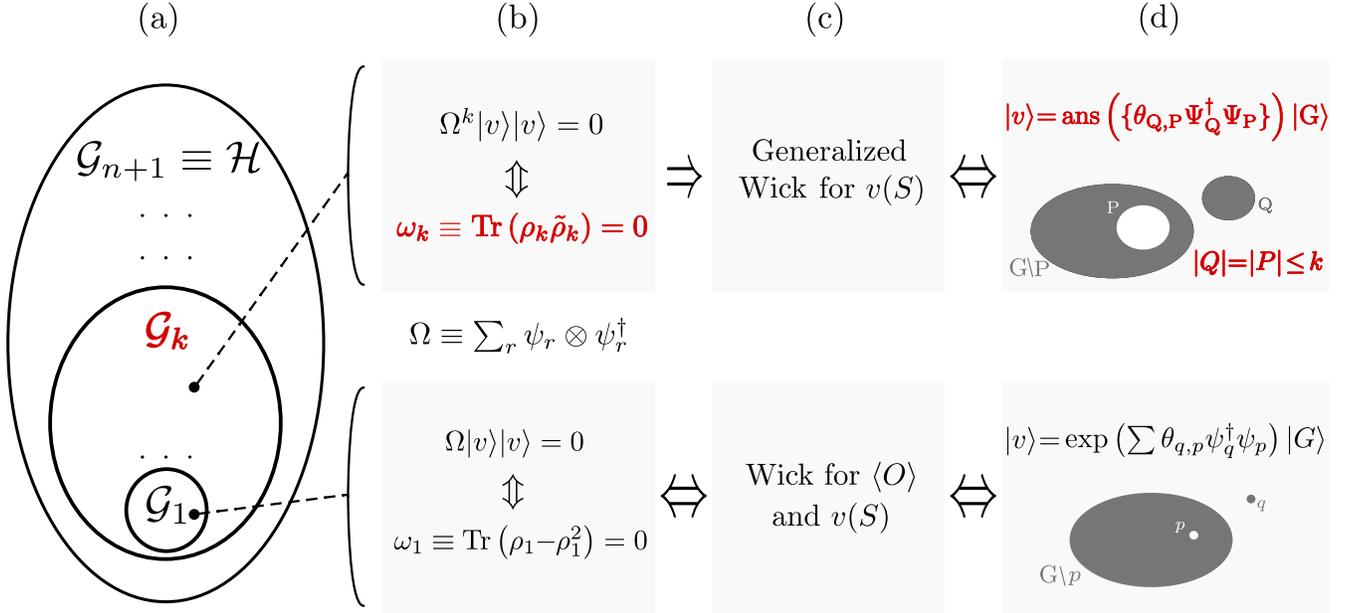}
\caption{\label{fig:vis_abs} Mathematical structure of our results. 
Highlighted in red are the main contributions: the twisted purity $\omega_k$, the set of correlated states $\GG_k$ defined by $\omega_k=0$, and the ansatz for states $\ket{v}\in\GG_k$. (a) The nested pattern of sets $\GG_k$. States in $\GG_1$ are the familiar Slater determinants. The entire Hilbert space of $n$ fermions on $l$ modes coincides with $\GG_{k=n+1}$. (b) $\omega_k$ generalizes the 1-RDM purity $\omega_1$ to the $k$-body case via `twisted' $k$-RDM $\tilde{\rho}_k$ (see Eq.\,\ref{eq:conj_rho_def}). Vanishing of twisted purity is equivalent to a generalization of Plücker relations $\Omega\ket{v}\ket{v}=0$. (c) As a key technical step, we find that states in $\GG_k$ obey a generalization of Wick's rule (Theorem\,\ref{thm:thm_corr_decomposition}) for the amplitudes $v(S)$, although not for observables $\langle O \rangle$. Unlike in the Slater case, there may be states outside $\GG_k$ which follow this generalized Wick's rule --- hence the `$\Rightarrow$' sign. (d) Generalized Wick's rule is equivalent to the ansatz representation for state $\ket{v}$. 
The explicit form of the ansatz is written in Eqs.\,\ref{eq:hyperferm_ansatz_maintext}-\ref{eq:ansatz_func_maintext}. The diagram displays how sets of modes $Q$ and $P$ relate to the set of modes $G$ occupied in the reference Slater state. The polynomial number of parameters $\theta_{P,Q}$ follows from the condition $|Q|=|P|\leq k$. For the Slater states $\ket{v}\in\GG_1$, the ansatz reduces to a known parameterization using an exponential of a weight-$2$ operator.
}
\end{figure*}

\textit{Generalized Plücker relations and twisted purity.}
In the non-Slater part of the Hilbert space, states can contain differing degrees of correlation.
This work is centered around states $\ket{v}$ with limited correlations, as defined by our generalization of Plücker relations: 
\begin{align}
    \Omega^{k} \ket{v}\otimes \ket{v}=0,
    \label{eq:k_Plücker}
\end{align}
where $k$ is a positive integer.
We denote $\GG_k\subset \mathcal{H}$ as the set of states satisfying Eq.\,\ref{eq:k_Plücker} --- in particular, $\GG_{k=1}$ are Slater (cf. Eq.\,\ref{eq:Plücker_rel}). From the structure of Eq.\,\ref{eq:k_Plücker}, we have $\GG_{k_1}\subset \GG_{k_2}$ if $k_1<k_2$, inducing a nested pattern of $\GG_{k}$ (Fig.\,\ref{fig:vis_abs}). Because $\Omega^{n+1}\ket{v}\otimes \ket{v}$ vanishes for any $n$-fermion state, $\GG_{k=n+1}\equiv \mathcal{H}$.

We now introduce twisted purity $\omega_k$, defined as
\begin{align}
    \label{eq:omega_def}
    \omega_k\left(\ket{v}\right)&\defeq\left|\frac{\Omega^{k} \ket{v}\otimes \ket{v}}{k!} \right|^2 ,
\end{align}
Equation \ref{eq:k_Plücker} is manifestly equivalent to $\omega_k$ being zero. The twisted purity is single-particle invariant and can be expressed as
\begin{align}
    \omega_k=\rm{Tr}[\rho_{k} \tilde{\rho}_{k}],
    \label{eq:omega_rho_def}
\end{align}
where $\rho_k$ is an order $k$-body reduced density matrix ($k$-RDM) with matrix elements
\begin{align}
    \rho_k^{Q,P}=\bra{v}\Psi^\dag_P\Psi^{\phantom{\dagger}}_Q\ket{v}
\end{align}
for sets $P$ and $Q$ ($|P|=|Q|=k$) and $\tilde{\rho}_k$ its twisted version, which we define as
\begin{align}
\label{eq:conj_rho_def}
    \tilde{\rho}_k^{Q,P}= \bra{v}\Psi^{\phantom{\dagger}}_Q\Psi^\dag_P\ket{v}.
\end{align}
The term `twisted' is a reference to `twist product' \cite{haah16},
implying a flipped multiplication order of $\Psi^{\phantom{\dagger}}_Q$ and $\Psi^\dag_P$ compared to $\rho_k^{Q,P}$. Twisted $k$-RDM $\tilde{\rho}_k$ can also be viewed as a particle-hole dual of ${\rho}_k$.
Furthermore, $\tilde{\rho}_k$ can be expressed in ordinary $k'$-RDMs for $k'\leq k$ by commuting $\Psi^{{\dagger}}_Q$ through $\Psi^{\phantom{\dag}}_P$ (see Appendix~\ref{app:twisted_RDM}). Because such $k'$-RDMs are marginals of the $k$-RDM, we have that $\tilde{\rho}_k$ and $\omega_k$ are functions of $\rho_k$. 

The twisted purity $\omega_k$ generalizes the single-body purity $\omega_1$ (Eq.\,\ref{eq:scalar_Plücker}) to $k\in \mathbb{N}$. Note that $\omega_k$ is different from the reduced purity $\rm{Tr}[\rho_{k}^2]$ of Ref.\,\cite{fran13}.
For example, $\rm{Tr}[\rho_{k}^2]$ reaching its maximal value is a criterion for Slater states \cite{majtey2016multipartite}, while vanishing of $\omega_k$ defines $\GG_k$, which is a broader class. Twisted purity is also distinct from k-RDM cumulants \cite{kutzelnigg99, hanauer12}. The cumulants enjoy the property of additive separability, which twisted purities are lacking. On the other hand, our main result – classes $\mathcal{G}_k$ and their connection to computational physics – do not have a known analogue for k-RDM cumulants.

Twisted purities are in principle observable experimentally if two copies of the studied state $\ket{v}$ can be produced at will. Indeed, $\omega_k$ is an expectation value of a $4k$-body Hermitian operator $(\Omega^\dag)^k\Omega^k$ on a state $\ket{v}\otimes \ket{v}$ (cf. Eq.\,\ref{eq:omega_def}).
Extracting $\rho_k$ by tomography provides another route to obtaining $\omega_k$ from an experiment.

\textit{Meaning of $\GG_k$ classification.}
To understand the type of correlations captured by twisted purity $\omega_k$, we study various examples of states in $\GG_k$. A broad class of examples is given by states $\ket{v}$ which obey the condition
\begin{align}
\forall S_1, S_2,~v(S_{1,2})\neq 0:~~\frac{1}{2}|S_1\triangle S_2|< k,
\label{eq:slater_diameter}
\end{align}
where $\triangle$ means symmetric difference.
In other words, occupation numbers in basis states $\ket{S}$ in support of such $\ket{v}$ are close to each other in Hamming distance.
Any $\ket{v}$ satisfying Eq.\,\ref{eq:slater_diameter} belongs to $\GG_k$,
\begin{align}
    \Omega^k &\ket{v}\otimes \ket{v}=\sum_{R\subset [l], |R|=k}\Psi_R\otimes\Psi^{\dag}_R \ket{v}\otimes \ket{v}\notag\\
    &=\sum_{\substack{S_1,S_2\subset [l]\\
    R\subset S_1 \bs S_2, \\ |R|=k}} v(S_1) v(S_2) \Psi_R\ket{S_1} \Psi^{\dag}_R\ket{S_2}=0.
    \label{eq:H^R_in_H_k}
\end{align}
The condition $R\subset S_1 \bs S_2$ appeared in the sum due to $\Psi_R\ket{S_1}$ ($\Psi^\dag_R\ket{S_2}$) vanishing unless $R\subset S_1$ ($R\cap S_2=\varnothing $). For such $S_{1,2}$ that $v(S_{1,2})\neq 0$ we have $|S_1 \bs S_2|=\frac{1}{2}|S_1\triangle S_2|<k$. 
Therefore, the sum over $R\subset S_1 \bs S_2$ for $|R|=k$ only contains zero elements, yielding Eq.\,\ref{eq:H^R_in_H_k}.

Information-theoretically, the condition in Eq.\,\ref{eq:slater_diameter} captures the locality of correlations, discriminating between Bell-like correlated states and GHZ-like. For $(l,n)=(8,4)$, state $\ket{v_1}=\frac{1}{\sqrt{2}}(\ket{1,2,3,4}+\ket{1,2,5,6})$ obeys Eq.\,\ref{eq:slater_diameter} for $k=3$, while $\ket{v_2}=\frac{1}{\sqrt{2}}(\ket{1,2,3,4}+\ket{5,6,7,8})$ only for $k=5$. Moreover, one can prove that the state $\ket{v_2}$ will not obey Eq.\,\ref{eq:slater_diameter} for $k=3$ in \textit{any} single-particle mode basis. This follows from $\ket{v_2}\notin\GG_3$, which can be checked by a direct computation: $\omega_3(\ket{v_2})\neq 0$. From the single-particle invariance of $\GG_k$, if $\ket{v}$ obeys Eq.\,\ref{eq:slater_diameter} in at least some single-particle basis, then $\ket{v}\in\GG_k$.

Counting the degrees of freedom remaining under Eq.\,\ref{eq:slater_diameter} for a given $k$ is a difficult task. Let us instead bound it from below. Consider a subset of such states, obeying a condition relative to a fixed basis state $\ket{S_0}$:
\begin{align}
\forall S, ~v(S)\neq 0: |S\triangle S_0|< k.
\label{eq:slater_radius}
\end{align}
Eq.\,\ref{eq:slater_radius} implies Eq.\,\ref{eq:slater_diameter}: from $|S_1\triangle S_0|< k$ and $|S_2\triangle S_0|< k$ one recovers $\frac{1}{2}|S_1\triangle S_2|< k$. States obeying Eq.\,\ref{eq:slater_radius} form a linear subspace of $\mathcal{H}$ of poly-sized dimension
\begin{align}
    \sum^{\lfloor{k/2}\rfloor}_{r=0}\begin{pmatrix}
    n \\
    r
\end{pmatrix}\begin{pmatrix}
    l-n \\
    r
\end{pmatrix}\sim n^{\lfloor{k/2}\rfloor}(l-n)^{\lfloor{k/2}\rfloor}
\label{eq:dimension_counting}
\end{align}
for $n,l\gg k$. This lower bounds the number of degrees of freedom in the set defined by Eq.\,\ref{eq:slater_diameter}, and therefore also in $\GG_k$. Later we upper bound the size of $\GG_k$ by counting the degrees of freedom in our non-Slater ansatz.

States which obey Eq.\,\ref{eq:slater_diameter} (and thus belong to $\GG_k$) naturally arise in perturbation theory truncated at finite order. Consider the Dyson series for a ground state $\ket{v}$ of $H=H_0+\textstyle\sum V$. Here $H_0$ is free-fermionic with a unique ground state $\ket{S_0}$, and $V$ are $4$-fermion Coulomb interactions. The perturbation theory is converging if $H_0$ is gapped (e.g., for a band insulator) and $V$ is sufficiently weak.
The $r$th order truncation $\ket{v}^{(r)}$ of $\ket{v}$ consists of terms proportional to $V^{r'}\ket{S_0}$ for $r'\leq r$. Since $V^{r}$ has up-to-$4r$-fermion terms, basis states $\ket{S}$ contributing to $\ket{v}^{(r)}$ differ from $\ket{S_0}$ in at most $|S\triangle S_0|=4r$. This implies Eq.\,\ref{eq:slater_radius} and therefore also Eq.\,\ref{eq:slater_diameter} with $k=4r+1$. 

More broadly than truncated perturbation series, Eq.\,\ref{eq:slater_radius} defines precisely the family of states known in computational quantum chemistry as configuration-interaction (CI) states \cite{crem13, hofm03}. For instance, $k=5$ corresponds to the CISD family, where SD stands for `singles and doubles'.  
From a quantum chemistry standpoint, our classes $\GG_k$ can be viewed as a single-particle invariant (orbital-invariant) generalization of CI truncations. 
Conversely, nonzero $\omega_k$ signals the irreducibility of a state to truncated configuration-interaction form.

\begin{figure}[t]
\includegraphics[width=\linewidth]{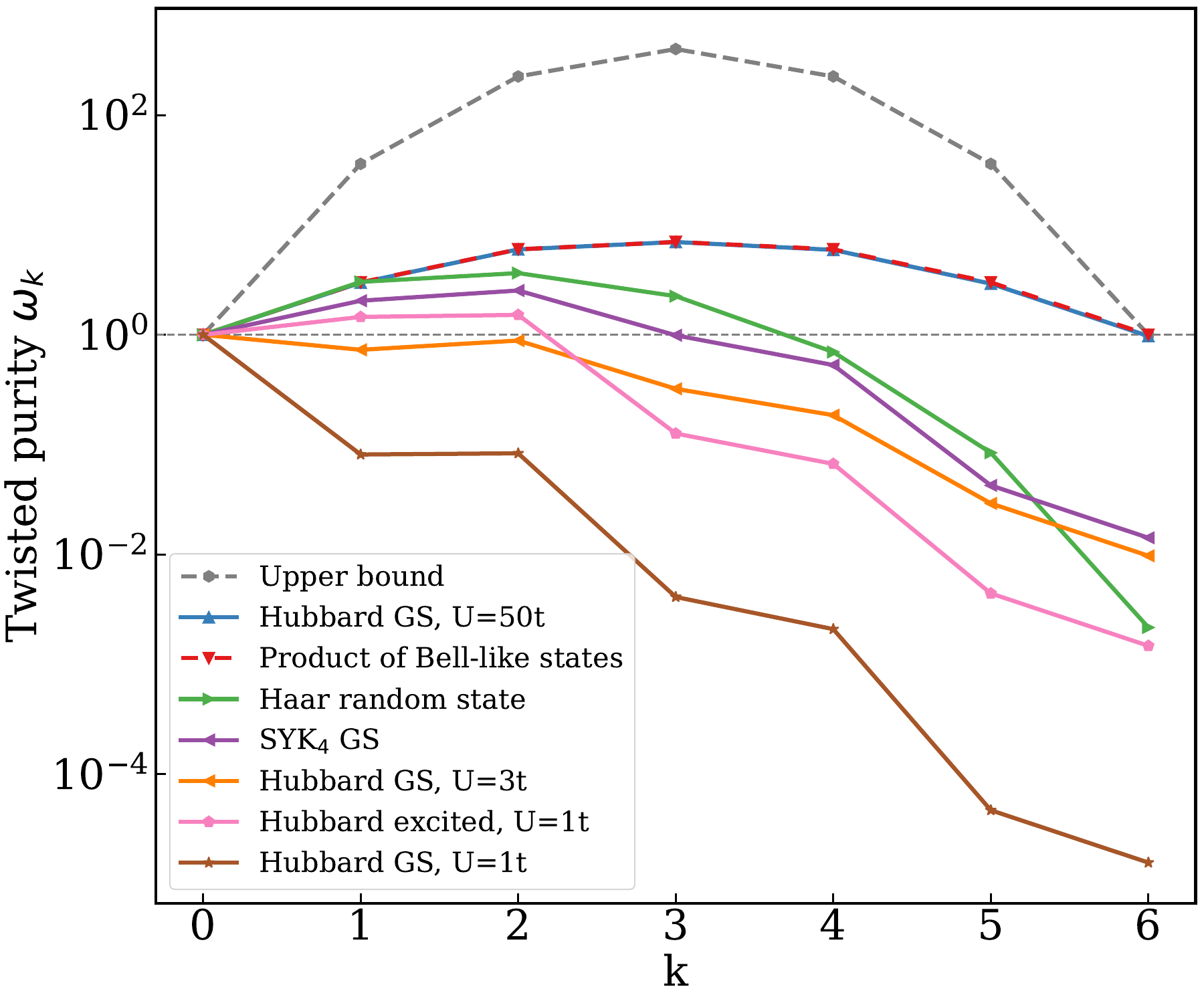}
\caption{\label{fig:numerics} 
Twisted purity $\omega_k$ for various $n=6$-fermionic states on $l=12$ modes; note the log scale. It is natural to include $\omega_0$ (equal to $1$ by Eq.\,\ref{eq:omega_def}) into the picture. For the Hubbard model eigenstates, the change of $\omega_k$ over $k$ depends on the relative coupling strength $U/t$ (cf. Eq.\,\ref{eq:Hubbard_model_def}) and excitation energy. For the ground state as a function of $U$, the state changes from being perturbatively close to $\GG_1$ ($U=t$) to being $\GG_3$-like ($U=3t$) 
to being far from any $\GG_k$ ($U=50t$). For $U=50t$, the curve is remarkably similar to the one of a product of Bell-like states. The pink curve shows an example excited state of the Hubbard model (of energy $\epsilon$ s.t. $(\epsilon - \epsilon_{\mathtt{min}})/(\epsilon_{\mathtt{max}} - \epsilon_{\mathtt{min}})) \sim 0.23$), with exponential decay starting from $k=3$. The ground state of a typical SYK model realization is highly correlated but with small $\omega_{k\geq 5}$, similarly to a Haar random state. Note the even-odd fluctutation in multiple plots. These come from the structure of $\omega_k$ itself; we hypothesize that Eq.\,\ref{eq:k_Plücker} for pairs $k=2r-1$ and $k=2r$ is in fact equivalent \cite{ftnt_oddeven}.} 
\end{figure}

To probe how $\omega_k$ manifests in strongly interacting physics, we studied numerically the twisted purities for eigenstates of 1D Hubbard model \cite{hubb64} at half-filling, as well as the complex SYK model \cite{kita15,sach93}. More information on these models and details of numerical simulations is given in Appendix~\ref{app:numerics_physics}; we published the code used to produce the data in the GitHub repository \cite{plucker_github}. The main results for $l=12$, $n=6$ are shown in Fig.\,\ref{fig:numerics}. For comparison, we include the twisted purities of two model states. These are the Haar-random state (here we plot the average $\omega_k$)
and the product of three Bell-like states $\ket{\varphi} = \tfrac{1}{\sqrt{2}}(\ket{1,2} + \ket{3,4})$. 
The twisted purities for these two cases can be computed analytically, see details in Appendices \ref{app:bell_therm_purity} and \ref{app:Haar_purity}. Also, we plot a basic upper bound on the values of $\omega_k$, derived in Appendix~\ref{app:bound} using Cauchy-Schwarz inequality and properties of $\rho_k$, $\tilde{\rho}_k$.

Our numerics suggests the relevance of $\GG_k$ classes for the Hubbard model eigenstates: we discover various eigenstates with small values of $\omega_k$ for $k\geq 1$, $3$, or $5$ (see Fig.\,\ref{fig:numerics}). States where $\omega_k$ exponentially decays after $k=1$ we interpret as perturbative. 
The more interesting states are those with non-monotonous twisted purity, e.g., when finite $\omega_{1,2}>1$ is followed by an exponential decay in $k$. Such states are qualitatively similar to $\GG_3$, but also non-perturbatively correlated (far from $\GG_1$). 
The ground state of the Hubbard model at a large coupling constant is far from any $\GG_k$. At the same time, we find that its twisted purities are in a perfect match with those of a product of Bell-like states. This suggests that this ground state might secretly have a product structure upon a global single-particle rotation. 

In our analysis, the smallness of twisted purity can be viewed as a simpler-to-compute substitute for the geometric proximity to $\GG_k$; however, we were unable to show an explicit bound on the distance to $\GG_k$ in terms of the magnitude of $\omega_k$. Such a bound can be proven for $k=1$ (for example, it follows from Lemma 2 in Ref.\,\cite{bittel2024optimal}), but this does not directly generalize to $k>1$.
Further characterization of states in and near $\GG_k$, especially those not falling under the condition of Eq.\,\ref{eq:slater_diameter}, is an interesting open question.

\textit{Extended Wick's rule and $\GG_k$ parameterization.}
To show that states $\ket{v}\in\GG_k$ admit an explicit non-Slater ansatz, the crucial step is showing a version of Wick's rule \cite{wick50, isserlis16} for amplitudes of $\ket{v}$. It is a generalization of Wick's rule for amplitudes that holds for Slater states $\ket{v}\in \GG_1$. To spell out both the Slater and non-Slater versions, organize basis states $\ket{S}$ as `excitations' with respect to a reference $\ket{G}$, namely $\ket{S}=\ket{G\cup Q \bs P }$ for sets $P\subset G$ and $Q\subset ([l]\backslash G)$, $|P|=|Q|$ \cite{ftnt_cup_bs}.
In a \textit{Slater} state $\ket{v}$, the amplitudes of `multi-particle excitations' (i.e., $|P|=|Q|>1$)  partition into single-particle excitations (assuming $v(G)\neq 0$),
\begin{align}
\frac{v(G\cup Q \bs P )}{v(G)}=(-1)^{|P|(|P|-1)/2}\det \left[\frac{v(G\cup q \bs p )}{v(G)}\right],
\label{eq:1_amplitude_Wick}
\end{align}
where the bracketed expression is a $|P|$-dimensional matrix over single modes $p\in P$ and $q\in Q$. Equation \ref{eq:1_amplitude_Wick} means that the state $\ket{v}$ is fully determined by its one-excitation amplitudes $v(G\cup q \bs p )$. This parallels the usual Wick's rule, which says that all observables on a Slater state are determined by single-particle correlators. Equation~\ref{eq:1_amplitude_Wick} can be derived from Wick's rule of
Ref.\,\cite{alexandrov13} (see Theorem 2.7; amplitude $v(G\cup Q \bs P )$ is proportional to $\bra{G}\Psi^{\phantom{\dag}}_Q \Psi^{\dag}_P\ket{v}$, and in particular, $v(G\cup q \bs p )$ are $2$-point correlators). Furthermore, Eq.\,\ref{eq:1_amplitude_Wick} is a special case of the extended Wick's rule that we lay out below.

To show an extension of Eq.\,\ref{eq:1_amplitude_Wick} to states in $\GG_k$, we examine the components of Eq.\,\ref{eq:k_Plücker} of type $\bra{A}\otimes \bra{B}\Omega^k\ket{v}\otimes\ket{v}$ for sets $A$ and $B$. These take the form:
\begin{align}
\notag
0=&\textstyle\sum_{R\subset(B\backslash A), |R|=k} \,v(A\cup R)\,v(B\bs R) \\
 &\times \sigma(A, R)\,\sigma(B\bs R, R).\label{eq:hi_Plücker_components}
\end{align}
Combining Eq.\,\ref{eq:hi_Plücker_components} with appropriately chosen $A$ and $B$, one arrives at the following relation, which holds for $|P|=|Q|>k$ (see Theorem\,\ref{thm:gen_wick_iter} in Appendix\,\ref{app:gen_wick_iter})

\begin{align}
\label{eq:wick_iter_maintext}
\frac{v(G\cup Q \bs P )}{v(G)}
&=
\sum_{\substack{P'\subsetneq P,~ Q'\subsetneq Q}}
\frac{v(G\cup Q' \bs P' )}{v(G)}\frac{v(G\cup \bar{Q}' \backslash \bar{P}')}{v(G)}~
\notag\\&\times 
(-1)^{|\bar{P}'|+1}\,\frac{|P'|}{|P|}\,\sigma(\bar{P}',P')\,\sigma(Q', \bar{Q}'),
\end{align}
where $\bar{P}'\defeq P\backslash P'$ and $\bar{Q}'\defeq Q\backslash Q'$. One observes that a multi-excitation amplitude $v(G\cup Q \bs P )$ is broken up into fewer-excitation amplitudes, as under the sum we have $|P'|,|\bar{P}'|<|P|$ and $|Q'|,|\bar{Q}'|<|Q|$.

For all terms on the right hand side of Eq.\,\ref{eq:wick_iter_maintext}, the factors $\frac{v(G\cup Q' \bs P' )}{v(G)}$ with $|P'|=|Q'|>k$ can be further decomposed using Eq.\,\ref{eq:wick_iter_maintext} again. This process can be continued iteratively, before one breaks $\frac{v(G\cup Q \bs P )}{v(G)}$ into $\frac{v(G\cup Q' \bs P' )}{v(G)}$ with $|P'|=|Q'|\leq k$ alone. 
The result is an extended version of the amplitude Wick's rule for Slater states $\ket{v}\in \GG_{k=1}$ (Eq.\,\ref{eq:1_amplitude_Wick}). In particular, it implies that ${v(G\cup Q' \bs P' )}$ for $|P'|=|Q'|\leq k$ fully characterize any $\ket{v}$ in the $\GG_k$ class! The explicit formula for the extended Wick's rule is complicated; we spell it out in Appendix~\ref{app:cumulant_exp} (see Theorem\,\ref{thm:thm_corr_decomposition}). For a reminiscent extension of Wick's theorem in the context of matrix product states, see Ref.~\cite{hubener13}.

Although the extended Wick's rule is a condition on an exponential number of amplitudes, we encode it entirely in a polynomially-sized ansatz for the whole state $\ket{v}$. This can be achieved via careful bookkeeping with the use of the generating function method (see Theorem~\ref{thm:ansatz} in Appendix~\ref{app:ansatz}). The result is that relative to a basis state $\ket{G}$ any state $\ket{v}\in \GG_k$ with $v(G)\neq 0$ takes the form 
\begin{align}
\ket{v}=v(G)\, F(T_1,..T_k)\ket{G},
\label{eq:hyperferm_ansatz_maintext}
\end{align}
for commuting nilpotent operators
\begin{align}
T_{k'}=\sum_{\substack{P\subset G, Q\subset [l]\bs{G},\\ |P|=|Q|=k'}}  \theta_{P,Q} \Psi^\dag_{Q} \Psi_{P},
\label{eq:Tk}
\end{align}
and the function 
\begin{align}
    F(x_1,\,&x_2,...) = \sqrt{1+ 2(x_2+x_4+..)}\notag\\   
    &\times \exp \left( \int\limits_0^{1} \dfrac{x_1+3x_3\mu^2+5x_5 \mu^4+..}{1+ 2(x_2 \mu^2 +x_4 \mu^4 +..)} d\mu \right).
    \label{eq:ansatz_func_maintext}
\end{align}
Complex numbers $\theta_{P,Q}$ quantify the violations of Eq.\,\ref{eq:wick_iter_maintext} for $|P|=|Q|\leq k$ (see Appendix~\ref{app:ansatz}). 
On the other hand, Eq.\,\ref{eq:hyperferm_ansatz_maintext} can be considered as an ansatz for the state $\ket{v}\in\GG_k$, in which case $\theta_{P,Q}$ are its unknown free parameters and amplitude $v(G)$ is coming from normalization. This ansatz has a polynomial size in the sense that the number of parameters is polynomial, growing asymptotically as $O(n^{k} (l-n)^k)$ (via counting similar to Eq.\,\ref{eq:dimension_counting}). As Eq.\,\ref{eq:hyperferm_ansatz_maintext} entirely covers $\GG_k$, this scaling upper bounds the number of degrees of freedom in $\GG_k$ itself.

\textit{Discussion.} 
An interesting open direction is to apply the ansatz of Eq.\,\ref{eq:hyperferm_ansatz_maintext} in a practical computation, such as the search for a ground state energy of a given model. 
One may employ the structure of Eq.\,\ref{eq:hyperferm_ansatz_maintext}, which is similar to the coupled-cluster ansatz --- except the latter uses $F(T_1,..T_k)=\exp(T_1+..+T_k)$. Because of this similarity, the numerical methods for coupled-cluster which don't rely on the exponential form of $F(T_1,..T_k)$, can be used with our ansatz (e.g., see Refs.\,\cite{degr16, thom10, spen16}). Another intriguing question is whether the structure of $k$-RDMs identified in this work can be used in the context of $N$-representability problem \cite{coleman63, klyachko06,liu07,mazziotti12}.

Using the formalism developed here, one may study the structure of states coming from a system with weak and \textit{sparse} interactions. For such states, one expects that for any $k$, all but a few components in $\Omega^k \ket{v}\otimes\ket{v}$ are extremely small. One may define a class of states where these small components are set to vanish exactly. For states in this class, the ansatz of the type given in Eq.\,\ref{eq:hyperferm_ansatz_maintext} may be compact due to many vanishing parameters $\theta_{P,Q}$. Unlike $\GG_{k}$, such a class would depend on the choice of the single-particle basis. Pinning down this structure is an interesting research direction.

An object worth further study is the generating function of twisted purities, $Z(\beta)=\sum_k \omega_k\beta^{2k}$. In this work, we found $Z(\beta)$ to be a useful analytical tool, owing to its multiplicativity for products of independent subsystems (see Appendices~\ref{app:bell_therm_purity},~\ref{app:Haar_purity}). However, the structure of this generating function could also have useful interpretations. For instance, its factorization gives a necessary condition that a state is a single-particle transformed product of small subsystems. Developing such criteria further could be viewed as a `size-consistent' \cite{popl76} extension of the present work. 

Other open research directions include extending our formalism to the states which are mixed or whose particle number is not fixed. Here, it may prove useful that operator $\Lambda$ of Refs.\,\cite{bravyi05, melo13} equals $\Omega+\Omega^\dag$ up to normalization. Finally, by analogy with the entanglement theory, a worthwhile goal is to investigate if twisted purities define a monotone under the free-fermionic operations, including measurements.

\textit{Acknowledgments.} We have benefited from discussions with X. Bonet-Monroig, V. Cheianov, P. Emonts, L. Ding, D. DiVincenzo, P. Gavrylenko, D. Gosset, J. Helsen, A. Izmaylov, J. Liebert, A. Lopez, J. Minar, T. Mori, T. O'Brien, S. Polla, A. Tikku, J. Zaanen, and Y. Zhang. The authors thank Barbara Terhal for giving detailed feedback on the manuscript. Y.H. acknowledges the hospitality of Perimeter Institute and that of the Bochenkova family, which provided conducive environments for focused work away from his home institutions. Research of M.S. at Perimeter Institute is supported in part by the Government of Canada through the Department of Innovation, Science and Economic Development and by the Province of Ontario through the Ministry of Colleges and Universities. Y.C. acknowledges supported by the funding from the Dutch Research Council (NWO) and the European Research Council (ERC) under the European Union’s Horizon 2020 research and innovation program. Y.H. is supported by QuTech NWO funding 2020-2024 – Part I “Fundamental Research”, project number 601.QT.001-1, financed by the Dutch Research Council (NWO), and the Quantum Software Consortium (NWO Zwaartekracht).

\appendix 
\onecolumngrid

\section{Fermionic algebra and sequence manipulations}
\label{app:ferm_seq}

Let $\mathcal{H}$ be the Hilbert space of $n$-fermion states on $l$ modes. Consider any fixed single-particle basis for the fermionic modes, defined by annihilation operators $\{\psi_j | j\in[l]\}$. Operators $\psi_j$ and their Hermitian conjugate creation operators $\psi^\dag_j$ obey the algebra
\begin{align}
    \psi_{j_1}\psi_{j_2}=-\psi_{j_2}\psi_{j_1}, ~\psi_{j_1}\psi^\dag_{j_2}+\psi^\dag_{j_2}\psi_{j_1}=\delta_{j_1 j_2},
\end{align}
where $\delta_{j_1 j_2}$ is the Kronecker delta. For an ordered sequence $S=(s_1,..s_n)$ ($0\leq s_1<..<s_n\leq l$), we define annihilation operator monomial $\Psi_S=\psi_{s_1}..\psi_{s_n}$. Notation $\Psi^\dag_S$ is defined as $\Psi^\dag_S=\left(\Psi_S\right)^\dag=\psi^\dag_{s_n}..\psi^\dag_{s_1}$ (note the change in multiplication order).
The basis $\{\ket{S}\}$ of $\mathcal{H}$ is defined as $\ket{S}=\psi^\dag_{s_n}..\psi^\dag_{s_1}\ket{\varnothing }=\left(\Psi_S\right)^{\dag}\ket{\varnothing }$. Here $\ket{\varnothing }$ is the Fock space vacuum.

For more convenient fermionic algebra manipulations, we introduce some formalities involving sequences of integers. We will only consider finite sequences; the length of a sequence $A=(a_1,..,a_{|A|})$ is denoted as $|A|$.  If all elements in $A$ are smaller than all elements in $B$ ($a<b,~\forall a\in A,~b\in B$), we say $A<B$. The ordered version of a sequence $A$ without repeating elements is denoted $\mathrm{ord}(A)$, 
i.e., $\mathrm{ord}(A)=(a_{\pi(1)},a_{\pi(2)}..a_{\pi(|A|)})$ for the ordering permutation $\pi:[|A|]\rightarrow [|A|]$, $a_{\pi(1)}< a_{\pi(2)}<..< a_{\pi(|A|)}$. A signature function $\sigma(A)=\pm 1$ for $A$ without repeating elements is equal to the sign of a permutation $\pi$ required to order $A$. We also define $\sigma(A,B)$ for a pair of sequences $A=(a_1, a_2,..a_{|A|})$ and $B=(b_1, ..b_{|B|})$ \textit{without shared or repeating elements}, as a signature for a single concatenated sequence $\sigma\left((a_1, a_2,..a_{|A|},b_1, ..b_{|B|})\right)$. Through sequence concatenation we also define $\sigma(A, B, C)$ for three arguments, as well as $\sigma$ for more arguments.

An ordered sequence $A=(a_1,..a_{|A|})$ without repeating elements naturally maps to a set, namely $\mathrm{set}(A)\defeq\{a_1,..a_{|A|}\}$. Vice-versa, a set of integers $A_{\mathrm{set}}=\{a_1,..a_{|A|}\}$ can be mapped to an ordered sequence $\mathrm{seq}(A_{\mathrm{set}})\defeq \mathrm{ord}\left((a_1,..a_{|A|})\right)$. This one-to-one mapping allows to use set-theoretic notions for ordered sequences without repeating elements. An intersection of two sequences is defined as $A\cap B = \mathrm{seq}(\mathrm{set}(A)\cap \mathrm{set}(B))$, same for union $\cup$, difference $\bs$ and symmetric difference $\triangle$. $A$ is a subsequence of $B$, or $A\subset B$, if $\mathrm{set}(A)\subset \mathrm{set}(B)$. For an integer $l$ with $[l]$ we denote the sequence $(1,2,..l)$, rather than a set $\{1,2,..,l\}$ as per usual. Thus adapting set-theoretic notation to sequences allows to give a simpler form to many technical parts of this work.

A few useful properties of the signature function $\sigma$ can be stated with help of set-theoretic notation. For ordered sequences without repeating or shared elements $A$, $B$ and $C$, we have (note $A\cup B = \mathrm{ord}\left((a_1,..,a_{|A|},b_1,..,b_{|B|})\right)$)
\begin{align}
    \sigma(A, B , C)&=\sigma(A, B)\sigma(A\cup B, C), \label{eq:sigma_rule_1}\\
    \sigma(A\cup B , C)&=\sigma(A, C)\sigma(B, C),\\
    \sigma(A, B)&=\sigma(B, A)(-1)^{|A|\cdot |B|},\\
    \sigma(A, B)&=1~\rm{if}~A<B.\label{eq:sigma_rule_4}
\end{align} 

As announced, the introduced formalities are convenient in dealing with fermionic algebra. For example
\begin{align}
    \Psi_A\Psi_B&=\sigma(A,B)\Psi_{A\cup B} \,\mathrm{for}\, A\cap B=\varnothing , \\
    \Psi^\dag_B \ket{A}&=\sigma(A, B) \ket{A\cup B}\,\mathrm{for}\, A\cap B=\varnothing , \\
    \Psi_B \ket{A}&=\sigma(A\bs B, B) \ket{A\bs B} \,\mathrm{for}\, B\subset A.
\end{align}

\section{Twisted RDM}
\label{app:twisted_RDM}
Here we obtain the decomposition of twisted k-RDM $\tilde{\rho}_k$ (Eq.\,\ref{eq:conj_rho_def}) into $k'$-RDM's $\rho_{k'}$ with $k'\leq k$. The key step is using relation $\psi_q\psi^\dag_p=-\psi^\dag_p\psi_q +\delta_{pq}$ in $\Psi_Q\Psi^\dag_P$ to permute all $\psi^\dag_p$ operators to the left. This will give rise to contractions $\Psi^\dag_{\bar{P}'}\Psi_{\bar{Q}'}\delta_{Q'P'}$ 
of all possible subsets $Q'\subset Q$ and $P'\subset P$ (denote $\bar{Q}'=Q\bs Q'$, same for $P$). To determine the sign in front of any such contraction, it is sufficient to consider an ad-hoc permutation 
\begin{align}
\Psi_Q\Psi^\dag_P&=\sigma(\bar{Q}',Q')\sigma(\bar{P}',{P}')\Psi_{\bar{Q}'}\Psi_{{Q}'}\Psi^\dag_{{P}'}\Psi^\dag_{\bar{P}'}\notag\\&=(-1)^{|{P}'|}\sigma(\bar{Q}',{Q}')\sigma(\bar{P}',{P}')\delta_{Q'P'}\Psi_{\bar{Q}'}\Psi^\dag_{\bar{P}'}+\mathrm{(..)}\notag\\&=(-1)^{|{P}'|+|\bar{P}'|}\sigma(\bar{Q}',{Q}')\sigma(\bar{P}',{P}')\delta_{Q'P'}\Psi^\dag_{\bar{P}'}\Psi_{\bar{Q}'}+\mathrm{(..)}
\end{align}
where $\mathrm{(..)}$ contains other contractions. By definitions of the matrix elements $\tilde{\rho}_k^{Q,P}=\bra{v}\Psi_Q\Psi^\dag_P\ket{v}$ (Eq.\,\ref{eq:conj_rho_def}) and $\rho_k^{Q,P}=\bra{v}\Psi^\dag_P\Psi^{\phantom\dag}_Q\ket{v}$  we then obtain ($|P|=|Q|=k$):
\begin{align}
    \tilde{\rho}_k^{Q,P}&=\sum_{P'\subset P,Q'\subset Q}(-1)^{|{P}'|+|\bar{P}'|}\sigma(\bar{Q}',{Q}')\sigma(\bar{P}',{P}')\delta_{Q'P'}\rm{Tr}[\Psi^\dag_{\bar{P}'}\Psi_{\bar{Q}'}\ket{v}\bra{v}] \notag
    \\&=\sum_{\substack{P'\subset P,Q'\subset Q\\ |P'|=|Q'| }}(-1)^{|P|}\sigma(\bar{Q}',{Q}')\sigma(\bar{P}',{P}')\delta_{Q'P'}\rho_{k=|\bar{P}'|}^{\bar{Q}'\bar{P}'}.
\end{align}
When sets $\bar{P}', \bar{Q}'$ are empty, we define $\rho_{k=|\bar{P}'|}^{\bar{Q}'\bar{P}'}\equiv \rm{Tr}[\ket{v}\bra{v}]=1$ .

\section{Numerical studies of Hubbard and SYK models}
\label{app:numerics_physics}

To construct the plots given in Fig.\,\ref{fig:numerics}, we consider the Hubbard model for spinful fermions on a chain with periodic boundary conditions. The single-particle operators are denoted $\psi_{x,\sigma}$ and $\psi^{\dag}_{x,\sigma}$ where $\sigma\in\{\uparrow,\downarrow\}$, $x = 1,...,l/2$ and $\psi_{l/2+1,\sigma} \defeq \psi_{1,\sigma}$, $\psi^\dag_{l/2+1,\sigma} \defeq \psi^\dag_{1,\sigma}$. Parameter $l$ is always even, so there are $l$ modes in total. For our simulations we used $l = 12$. The Hamiltonian of the model is
\begin{align}
H=&\textstyle\sum\limits_{x=1}^{l/2}\sum\limits_{\sigma\in\{\uparrow,\downarrow\}} \left( -t \psi^\dag_{x, \sigma} \psi_{x+1, \sigma}-t \psi^\dag_{x+1, \sigma} \psi_{x, \sigma}\right) + \sum\limits_{x=1}^{l/2} U n_{x, \uparrow}n_{x, \downarrow} \,,
\label{eq:Hubbard_model_def}
\end{align}
where $n_{x, \sigma}=\psi^\dag_{x, \sigma} \psi_{x, \sigma}$. We consider the model to be at half-filling $n=l/2$.\\

We have also studied the complex SYK model. The Hamiltonian of the model on $l$ sites is
\begin{equation}
    H = \sum_{a>b>c>d=1}^{l} (t_{abcd} \psi_a^{\dag}\psi_b^{\dag}\psi_c \psi_d + h.c.)
\end{equation}
with the random couplings, normally distributed according to
\begin{equation}
    \langle t_{abcd} \, t^*_{abcd} \rangle = \frac{1}{(2l)^{3}}, ~~~
    \langle (t_{abcd})^2 \rangle = \langle (t^*_{abcd})^2 \rangle = 0. 
\end{equation}
The data for the ground state of the SYK model given in Fig.\,\ref{fig:numerics} was produced for $l=12$ at half-filling $n=6$.\\

We have used the {\tt QuSpin} Python package to run the numerics. More about the library can be found on the web-page $\mathtt{http://quspin.github.io/QuSpin/}$. The eigenstates of Hamiltonians have been found using the exact diagonalization method. We have tried several methods to compute twisted purities $\omega_k$. The most efficient method which we found for small values of $k$ was by calculating the matrix elements of RDMs. For large values of $k$ it proved more efficient to sum over the squared generalized Plücker relations.

There is an alternative method to compute the whole set of twisted purities directly, using the definition given in Eq.\,\ref{eq:omega_def}. Once the operator $\Omega$ and the tensor square of the state are constructed, one can act multiple times with $\Omega$ on the tensor square. Each subsequent action takes the same time to compute, and each iteration allows to extract the new $\omega_k$ by taking the norm squared of the result. However, since one has to construct and store a tensor square of the state, the method requires extensive memory resources. Therefore, we opted for the methods described in the previous paragraph.

The Python code written in the course of this numerical investigation, as well as the generated datasets, are available from the corresponding author on request.
\section{Products of Bell-like states and their thermodynamic limits}
\label{app:bell_therm_purity}

We call the state
\begin{equation}
    \ket{\varphi} = \tfrac{1}{\sqrt{2}} (\ket{1,2} + \ket{3,4})
\end{equation}
a Bell-like fermionic state. It is one of the simplest states for which $\omega_1 \neq 0$. In this Appendix we compute twisted purities of a state which is a high tensor power of $\ket{\varphi}$. It proves convenient to use the generating function
\begin{equation}
    \label{eq:partitionfunc}
    Z(\ket{v}, \beta)=\sum_{k=0}^{l} \omega_k(\ket{v}) \beta^{2k}
\end{equation}
because of its multiplicative properties. Let $\ket{v} = \ket{v_1} \wedge \ket{v_2}$ be the product state of fermionic system containing two subsystems. In this case $\Omega = \Omega_1 + \Omega_2$, $[\Omega_1 , \Omega_2] = 0$. Noting that
\begin{equation}
    \sum_k \omega_k\beta^{2k} = |e^{\beta \Omega} \ket{v} \otimes \ket{v}|^2.
\end{equation}
one can deduce
\begin{align}
    \notag Z(\ket{v_1} \wedge \ket{v_2}, \beta) = |e^{\beta \Omega} \ket{v} \otimes \ket{v}|^2 = |e^{\beta (\Omega_1 + \Omega_2)} (\ket{v_1}\wedge \ket{v_2}) \otimes (\ket{v_1} \wedge \ket{v_2})|^2 = \\ 
    = |(e^{\beta \Omega_1} \ket{v_1} \otimes \ket{v_1}) \otimes (e^{\beta \Omega_2} \ket{v_2} \otimes \ket{v_2})|^2 = Z(\ket{v_1},\beta) Z(\ket{v_2},\beta).
\end{align}
Now consider a state $(\ket{\varphi})^{\wedge p}$ on $l = 4p$ sites and with $n = 2p$ particles. A direct computation shows that $Z(\ket{\varphi},\beta) = 1 + \beta^2 + \beta^4$, so
\begin{equation}
    Z(\ket{\varphi}^{\wedge l/4}, \beta) = (1+\beta^2 + \beta^4)^{l/4}.
\end{equation}
For example, for $l=12$ this equation gives $Z = 1 + 3 \beta^2 + 6 \beta^4 + 7 \beta^6 + 6 \beta^8 + 3 \beta^{10} + \beta^{12}$. \\

We use the multiplicativity of $Z$ to identify the properties of $\omega_k$ for locally-correlated systems in the thermodynamic limit. In the process, we will obtain a version of central limit theorem. By `locally-correlated' we mean that the system can be broken into a set of subsystems, whose sizes are small compared to the whole, while the correlations between subsystems can be neglected. Let $Z_{0}(\beta^2) = \sum_{k=1}^{r} \beta^{2k} \omega_{0,k}$ be the generating function of twisted purities for the state of one small subsystem. Generically, it will be a polynomial in $\beta^2$ of degree $r = \rho d$, where $d \ll l$ is the number of fermionic sites in the subsystem and $\rho = n/l$ is the filling factor of the system. We assume $\rho < 1/2$. The generating function for the whole system would be
\begin{equation}
    Z(\beta^2) = (Z_0(\beta^2))^{\tfrac{l}{d}} = \prod_{i=1}^{r} (1 + p_i \beta^2)^{\tfrac{l}{d}}.
\end{equation}
for some coefficients $p_i$, coming from factorization of the polynomial $Z_0(\beta^2) = \prod_{i=1}^{r} (1 + p_i \beta^2)$ with the property $Z_0(0) = 1$. Using the binomial and Stirling formulas one can find that each term behaves in $l \to \infty$ limit as
\begin{equation}
    (1 + p_i \beta^2)^{\tfrac{l}{d}} \sim \sum_k \beta^{2k} a_i \exp \left(l b_i - l c_i(k/l - \varkappa_i)^2\right),
\end{equation}
\begin{equation}
    a_i = \dfrac{1+p_i}{\sqrt{2 \pi l}}\sqrt{\frac{d}{p_i}}, ~~~ b_i = \frac{\log[1+p_i]}{d}, ~~~ c_i = \dfrac{d(1+p_i)^2}{2 p_i}, ~~~ \varkappa_i = \dfrac{p_i}{d(1+p_i)}.
\end{equation}
Now, multiplying the expressions with the different $p_i$, one finds total generating function to be
\begin{equation}
    Z(\beta^2) \sim \sum_k \beta^{2k} \tilde{a} \exp \left(l \tilde{b} - l \tilde{c}(k/l - \tilde{\varkappa})^2\right), ~~~ \text{where} ~~~ \tilde b = \sum_{i=1}^{r} b_i, ~~~ \dfrac{1}{\tilde{c}} = \sum_{i=1}^{r} \dfrac{1}{c_i}, ~~~ \tilde{\varkappa} = \sum_{i=1}^{r} \varkappa_i.
\end{equation}
Taking logarithmic derivatives of $Z_0(\beta^2)$, the parameters can be expressed in terms of the purities of subsystems
\begin{equation}
    e^{l\tilde{b}} = \left(\sum_{k=1}^{r} \omega_{0,k} \right)^{l/d}, ~~~
    \tilde{\varkappa} = \dfrac{1}{d} \dfrac{\sum_{k=1}^{r} k \omega_{0,k}}{\sum_{k=1}^{r} \omega_{0,k}}, ~~~
    \dfrac{1}{\tilde{c}} = \dfrac{2}{d} \dfrac{\sum_{k=1}^{r} k^2 \omega_{0,k}}{\sum_{k=1}^{r} \omega_{0,k}} - \dfrac{2}{d} \left( \dfrac{\sum_{k=1}^{r} k \omega_{0,k}}{\sum_{k=1}^{r} \omega_{0,k}}\right)^2.
\end{equation}
Thus we find that in thermodynamic limit, only a few parameters control the dominating contributions to $\omega_k$.

It is interesting to note that a tensor power of Bell states does not belong to a class $\mathcal{G}_k$ for a finite $k$, despite consisting of disentangled finite blocks. This example demonstrates that strong fermionic correlations signalled by nonzero $\omega_{k=O(n)}$, similarly to extensive fermionic magic \cite{cudby23, reardon23, dias23}, are compatible with an absence of system-wide entanglement. It is an intriguing open question, whether entanglement and $k$-body fermionic correlations can be simultaneously captured with a single measure in a useful way. In our view, the answer to this question may well be negative. It also connects to the problem of constructing size-consistent fermionic correlation measures, discussed in the main text.

\section{Purities of Haar random states}
\label{app:Haar_purity}
In this section we will compute average twisted purities of the real random states $\ket{v} = \sum_{A \in [l]^k} v_{A} \ket{A}$, $v_{A} \in \mathbb{R}$, in a Hilbert space of $n$ fermions on $l$ sites. We consider the states to be distributed uniformly, i.e. the states are given by a random points on $S^d$, $d = {l \choose n}$, with the $SO(d)$ invariant distribution. We will refer to these states as to (real) Haar states, and denote the average over this measure by $\langle ~~~ \rangle_{|v|^2=1}$. The averaged purity is
\begin{equation}
    \langle \omega_k \rangle  = \dfrac{1}{(k!)^2} \left\langle \bra{v} \otimes \bra{v} (\Omega^{\dag})^{k} \Omega^k \ket{v} \otimes \ket{v} \right\rangle _{|v|^2=1}.
\end{equation}
Because of $SO(d)$ invariance, the four-point correlation function $\langle v_A v_B v_C v_D \rangle_{|v|^2=1}$ should be some quadratic combination of Kronecker $\delta$-symbols. Since it is symmetric under permutation of indices, up to overall factor it should be equal to $\delta_{AB} \delta_{CD} + \delta_{AC} \delta_{BD} + \delta_{AD} \delta_{BC}$. The normalization factor can be fixed contracting the correlation function with $\delta_{AB} \delta_{CD}$, and using that $\langle 1 \rangle_{|v|^2=1} = 1$. This results in
\begin{equation}
\label{eq:4ptHaar}
    \langle v_A v_B v_C v_D \rangle_{|v|^2=1} = \dfrac{1}{d(d+2)}(\delta_{AB} \delta_{CD} + \delta_{AC} \delta_{BD} + \delta_{AD} \delta_{BC}).
\end{equation}
which gives for the averaged twisted purities
\begin{align}
    \langle \omega_k \rangle = & \dfrac{1}{d(d+2)} \sum_{|I|,|J|=k} \sum_{|A|,|B|=n} \left(  \langle A | \Psi_I \Psi^{\dag}_J | A \rangle \langle B | \Psi^\dag_I \Psi_J | B \rangle + \right. \\
    & \left. + \langle A | \Psi_I \Psi^{\dag}_J | B \rangle \langle B | \Psi^\dag_I \Psi_J | A \rangle + \langle A | \Psi_I \Psi^{\dag}_J | B \rangle \langle A| \Psi^\dag_I \Psi_J | B \rangle \right).
\end{align}
Three terms in the sum can be computed as
\begin{equation}
    \sum_{|I|,|J|=k} \sum_{|A|,|B|=n} \langle A | \Psi_I \Psi^{\dag}_J | A \rangle \langle B | \Psi^\dag_I \Psi_J | B \rangle = \sum_{|I|,|J|=n} \sum_{|A|,|B|=K} \delta_{I=J}\delta_{I \cap A = \varnothing } \delta_{J \subset B},
\end{equation}
\begin{equation}
    \sum_{|I|,|J|=k} \sum_{|A|,|B|=n} \langle A | \Psi_I \Psi^{\dag}_J | B \rangle \langle B | \Psi^\dag_I \Psi_J | A \rangle = \sum_{|I|,|J|=k} \sum_{|A|,|B|=n} \delta_{A\cup I=B\cup J} \delta_{J \cap B = \varnothing } \delta_{I \cap A = \varnothing } \delta_{A \backslash J = B \backslash I} \delta_{J \subset A} \delta_{I \subset B},
\end{equation}
\begin{equation}
    \sum_{|I|,|J|=k} \sum_{|A|,|B|=n} \langle A | \Psi_I \Psi^{\dag}_J | B \rangle \langle A| \Psi^\dag_I \Psi_J | B \rangle = 0,
\end{equation}
which gives after summation
\begin{equation}
    \label{eq:Haar}
    \langle \omega_k \rangle = \dfrac{ {l \choose k} {l-k \choose n} {l-k \choose n-k} + {l \choose n-k} {l-n \choose k} {l-n+k \choose k}}{ {l \choose n}\left( {l \choose n} + 2 \right) } = \frac{1}{{l \choose n} + 2}\dfrac{k! (l-k)! + n! (l-n)!}{(k!)^2 (n-k)! (l-n-k)!}.
\end{equation}

Since this formula is exact, we can see explicitly how $\langle\omega_k\rangle$ gets concentrated around its thermodynamically preferred value at $l, n \gg 1$. Using the Stirling formula $p! \sim \sqrt{2\pi p}(p/e)^p$ one can estimate the relation of contributions of the first and second terms in Eq.\,\ref{eq:Haar} as
\begin{equation}
    \label{eq:formfactrel}
    \dfrac{k! (l-k)!}{n! (l-n)!} \sim e^{l(s(\varkappa) - s(\rho))}, ~~~ s(x) = \Lambda(x) + \Lambda(1-x), ~~~ \Lambda(x) = x \log (x)
\end{equation}
where $k = \varkappa l$, $n = \rho l$. For half-filled system one has $0<\varkappa<\rho < 1/2$ and $s(x)$ is monotonically decreasing, so $s(\varkappa)>s(\rho)$. Thus the second term in (\ref{eq:Haar}) is exponentially suppressed. Now the dominating value of $k$ can be found by extremizing the function
\begin{equation}
    \dfrac{\log \langle \omega_k \rangle}{l} \sim \Lambda(\rho) + \Lambda(1-\rho) + \Lambda(1-\varkappa) - \Lambda(\rho - \varkappa)  - \Lambda(\varkappa) - \Lambda(1 - \rho - \varkappa)
\end{equation}
with respect to $\varkappa$. This gives:
\begin{equation}
    \langle \omega_k \rangle \sim a\exp \left( l b - l c(k/l-\varkappa_*)^2 \right)
\end{equation}
where
\begin{equation}
    2\varkappa_* (1 - \varkappa_*) = \rho(1-\rho) ~~~ \Rightarrow ~~~ \varkappa_* = \frac{1}{2}(1 - \kappa), ~~~ \kappa = \sqrt{1-2\rho(1-\rho)},
\end{equation}
\begin{equation}
    a = \dfrac{1}{\sqrt{\pi l}}\sqrt{\dfrac{1+\kappa}{1-\kappa}}, ~~~ c = \dfrac{2 \sqrt{1-2\rho(1-\rho)}}{\rho(1-\rho)},
\end{equation}
\begin{equation}
    b = \Lambda(1-\rho)+\Lambda(\rho)+ \frac{1}{2} \Lambda(1+\kappa) - \frac{1}{2} \Lambda(1-\kappa) - \frac{1}{2} \Lambda(\kappa-2\rho + 1) - \frac{1}{2} \Lambda(\kappa+2\rho-1).
\end{equation}

It is curious to note the simple behaviour of $\langle \omega_k \rangle$ at $k \ll n < l$
\begin{equation}
    \langle \omega_k \rangle \sim \dfrac{(\rho  (1-\rho))^k l^k}{k!} = \dfrac{(\langle\omega_1 \rangle)^k}{k!}.
\end{equation}
In this case the generating function from Eq.\,\ref{eq:partitionfunc} behaves as
\begin{equation}
    \langle Z(\beta) \rangle \sim e^{ \beta^2 \langle \omega_1 \rangle}
\end{equation}
which is consistent with multiplicativity under the disjoint union of the systems because $\langle\omega_1\rangle$ is additive.

\section{Upper bound on the twisted purities}
\label{app:bound}
The values of the twisted purities can be bounded applying the Cauchy-Schwarz inequality for the operators
\begin{equation}
    \label{eq:CauchySchwarz}
    \Tr[A^{\dag} B] \leq \sqrt{\Tr[A^\dag A] \Tr[B^\dag B]}.
\end{equation}
Using that both k-RDM and twisted k-RDM are Hermitian operators $(\rho_k^{P,Q})^* = \rho_k^{Q,P}$, $(\tilde{\rho}_k^{P,Q})^* = \tilde{\rho}_k^{Q,P}$ one finds
\begin{equation}
\omega_k = \Tr[\tilde{\rho}_k \rho_k] \leq \sqrt{\Tr[\tilde{\rho_k}^\dag \tilde{\rho}_k] \Tr[\rho_k^\dag \rho_k]} = \sqrt{\Tr[\tilde{\rho}_k^2] \Tr[\rho_k^2]}.
\end{equation}
The operator $\rho_k$ is positive, i.e. all of its eigenvalues are $\geq 0$, which gives for its trace $\Tr[\rho_k^2] \leq (\Tr[\rho_k])^2$. To compute the traces of k-RDM
\begin{equation}
    \Tr[\rho_k] = \sum_{P \subset [l], \, |P|=k} \rho_k^{P,P} = \sum_{0\leq i_k<...<i_1 \leq l} \bra{v} \psi_{i_k}^{\dag} ... \psi_{i_1}^{\dag} \psi_{i_1} ... \psi_{i_k} \ket{v} = \frac{1}{k!}\sum_{i_k\neq ... \neq i_1} \bra{v} n_{i_k} ... n_{i_1} \ket{v}
\end{equation}
note that the boundary terms in a sum can be resolved as
\begin{equation}
    \sum_{i_k\neq ... \neq i_1} = \sum_{i_k}\sum_{i_{k-1}\neq ... \neq i_1} - \sum_{1 \leq \alpha \leq k} \sum_{i_{k-1} = i_{\alpha}\neq ... \neq i_1}.
\end{equation}
Together with $\sum_{0\leq i \leq l} {n_i} = n$ and $(n_i)^2 = n_i$ this gives a recursive formula
\begin{equation}
    \sum_{i_k\neq ... \neq i_1} \bra{v} n_{i_k} ... n_{i_1} \ket{v} = (n-k+1) \sum_{i_{k-1}\neq ... \neq i_1} \bra{v} n_{i_{k-1}} ... n_{i_1} \ket{v} = ... = \frac{n!}{(n-k)!}.
\end{equation}
Similar consideration is applicable to $\tilde{\rho}_k$ with $n_i$ being replaced by $1-n_i$ and $n$ being replaced by $l-n$. Collecting all the formulas together one gets a bound
\begin{equation}
    \omega_k \leq  {n \choose k} {l-n \choose k}.
\end{equation}

\section{Derivation of extended Wick's rule (recursive form)}

\label{app:gen_wick_iter}

\begin{theorem}
\label{thm:gen_wick_iter}
Let $\ket{v}=\sum_{S\subset [l], |S|=n} v(S)\ket{S}$ be a state in $\mathcal{H}$ such that $\Omega^k \ket{v}\ket{v}=0$. Consider three ordered sequences $G=(1,..n)$, $P\subset G$ and $Q\subset [l]\bs{G}$. If $|Q|=|P|> k$,
\begin{align}
\frac{v(G\cup Q \bs P )}{v(G)}
=
\sum_{\substack{P'\subsetneq P,~ Q'\subsetneq Q}}
(-1)^{|\bar{P}'|+1}\,\frac{|P'|}{|P|}\,\sigma(\bar{P}',P')\,\sigma(Q',\bar{Q}') ~\frac{v(G\cup Q' \bs P' )}{v(G)} \frac{v(G\cup \bar{Q}' \bs \bar{P}' )}{v(G)}
\label{eq:gen_wick_iter}.
\end{align}

\end{theorem}
\begin{proof}
Since $\Omega^k \ket{v}^{\otimes 2}=0$ implies $\Omega^{k'} \ket{v}^{\otimes 2}=0$ for all $k'> k$, it is sufficient to prove Eq.\,\ref{eq:gen_wick_iter} for $P$ and $Q$ such that $|P|=|Q|=k+1$. Consider the component $\bra{A}\bra{B}\Omega^k \ket{v}\otimes \ket{v}=0$ (Eq.\,\ref{eq:hi_Plücker_components}):
\begin{align}
&\textstyle\sum_{R\subset(B\backslash A), |R|=k} \sigma(A, R)\,\sigma(B\bs R, R)~ v(A\cup R)\,v(B\bs R)=0,
\label{eq:hi_Plücker_components_repeated}
\end{align}
choosing $A=(G\bs P) \cup \tilde{p}$ and $B=G\cup Q \bs \tilde{p}$ for some fixed $\tilde{p}\in P$. 
Observe that $R$ is a subset of $(B\backslash A)= Q \cup P \bs \tilde{p}$ with length $|R|=k$, while $|Q|=|P|=k+1$. This implies that $R$ can be represented as $R=\left(P'\cup (Q\setminus Q')\right) \bs \tilde{p}$ for certain $P'\subset P$ and $Q'\subset Q$ such that $\tilde{p}\in P'$ and $|Q'|=|P'|$. We also introduce simplified notations $\bar{Q}'=Q\bs Q'$ and $\bar{P}'=P\bs P'$. The goal of these substitutions is to give the amplitude product $v(A\cup R)\,v(B\bs R)$ the form $v(G\cup Q'\bs P')\,v(G\cup \bar{Q}'\bs \bar{P}')$, bringing Eq.\,\ref{eq:hi_Plücker_components_repeated} closer to the desired structure of Eq.\,\ref{eq:gen_wick_iter}.

With the above definitions for $A$, $B$, $R$ and noting that $Q>G$ and $Q>P$ by assumptions of Theorem\,\ref{thm:gen_wick_iter}, let us restructure the sign factor $\sigma(A, R)\,\sigma(B\bs R, R)$:
\begin{align}
\notag \sigma((G\bs P) \cup \tilde{p},R) &~\sigma((G\cup Q\setminus \tilde{p}\setminus R) , R )
=\notag \sigma((G\bs P) \cup \tilde{p},P'\bs \tilde{p}, \bar{Q}' ) ~\sigma((G\cup Q\setminus P '\bs \bar{Q}' ) , P'\bs \tilde{p}, \bar{Q}'  )\\
\notag 
&=\sigma((G\bs P) \cup \tilde{p},P'\bs \tilde{p}) ~\sigma(G\bs P ', Q', P'\bs \tilde{p}, \bar{Q}' )\\
\notag 
&=\sigma((G\bs P) \cup \tilde{p},P'\bs \tilde{p})~\sigma(G\bs P' , P'\bs \tilde{p}, Q',\bar{Q}')~(-1)^{|Q'| (|P'|-1)}\\
\notag 
&=\sigma((G\bs P) \cup \tilde{p},P'\bs \tilde{p})~\sigma(G\bs P' , P'\bs \tilde{p})~\sigma(Q',\bar{Q}')\\
\notag 
&=\sigma(G\bs P',P'\bs \tilde{p})~\sigma(\bar{P}' ,P'\bs \tilde{p})~\sigma(\tilde{p},P'\bs \tilde{p})~\sigma(G\bs P' , P'\bs \tilde{p})~\sigma(Q',\bar{Q}')\\
\notag 
&=\sigma(\bar{P}' ,P'\bs \tilde{p})~\sigma(\tilde{p},P'\bs \tilde{p})~\sigma(Q',\bar{Q}')\\
\notag 
&=\sigma(\bar{P}' ,P'\bs \tilde{p})~\sigma(\tilde{p},\bar{P}')~\sigma(\tilde{p},P\bs \tilde{p})~\sigma(Q',\bar{Q}')\\
\notag 
&=(-1)^{|\bar{P}'|}~\sigma(\bar{P}' ,P'\bs \tilde{p})~\sigma(\bar{P}',\tilde{p})~\sigma(\tilde{p},P\bs \tilde{p})~\sigma(Q',\bar{Q}')\\
&=(-1)^{|\bar{P}'|}~\sigma(\tilde{p},P\bs \tilde{p})~\sigma(\bar{P}' ,P')~\sigma(Q',\bar{Q}')
\label{eq:v_sign_factor}
\end{align}

The strategy in this derivation was to (i) separate $\sigma$'s involving $Q$ on the one hand and $G$ and $P$ on the other hand, then to (ii) eliminate the dependency of the expression on $G$, and finally to (iii) separate $\sigma$'s involving $P'$ and $\tilde{p}$. Since the factor $\sigma(\tilde{p}, P\bs \tilde{p})$ does not depend on the choice of $P'$ or $Q'$, it can be eliminated from Eq.\,\ref{eq:hi_Plücker_components_repeated} altogether, yielding (recall substitutions $A=(G\bs P) \cup \tilde{p}$, $B=G\cup Q \bs \tilde{p}$, and $R=(P'\cup \bar{Q}') \bs \tilde{p}$):

\begin{align}
\sum_{\substack{P'\subset  P,~ Q'\subset Q\\\rm{:}\,\tilde{p}\in P' \& |P'|=|Q'|}}&  v(G\cup \bar{Q}'\bs \bar{P}')~v(G\cup Q'\bs P' )~(-1)^{|\bar{P}'|}~\sigma(\bar{P}', P') ~\sigma(Q',\bar{Q}')=0.
\label{eq:recur_v_tilde_together}
\end{align}

Moving the terms corresponding to $P'\neq P$, $Q'\neq Q$ to the right-hand side, we obtain

\begin{align}
v(G)~v(G\cup Q\setminus P) =\sum_{\substack{P'\subsetneq P,~ Q'\subsetneq Q\\\rm{:}\,\tilde{p}\in P' \& |P'|=|Q'|}}&  v(G\cup \bar{Q}'\bs \bar{P}')~v(G\cup Q'\bs P' )~(-1)^{|\bar{P}'|+1}~\sigma(\bar{P}', P') ~\sigma(Q',\bar{Q}').
\end{align}

Dividing both sides by $v(G)^2$ yields

\begin{align}
\frac{v(G\cup Q \bs P )}{v(G)} =\sum_{\substack{P'\subsetneq P,~ Q'\subsetneq Q\\\rm{:}\,\tilde{p}\in P' \& |P'|=|Q'|}}&  \frac{v(G\cup Q' \bs P' )}{v(G)}\frac{v(G\cup \bar{Q}' \bs \bar{P}' )}{v(G)}~(-1)^{|\bar{P}'|+1}~\sigma(\bar{P}', P') ~\sigma(Q',\bar{Q}').
\label{eq:recur_v_tilde_separated}
\end{align}

The theorem statement (Eq.\,\ref{eq:gen_wick_iter}) is then obtained by summing up Eq.\,\ref{eq:recur_v_tilde_separated} for all $\tilde{p}\in P$:

\begin{align}
\notag|P|~\frac{v(G\cup Q \bs P )}{v(G)}
= &\sum_{\tilde{p}\in P}
\sum_{\substack{P'\subsetneq P,~ Q'\subsetneq Q\\ \rm{:}\,\tilde{p}\in P'}}
\frac{v(G\cup Q' \bs P' )}{v(G)}\frac{v(G\cup \bar{Q}' \bs \bar{P}' )}{v(G)}~(-1)^{|\bar{P}'|+1}~\sigma(\bar{P}', P') ~\sigma(Q',\bar{Q}')\\
=&
\sum_{\substack{P'\subsetneq P,~ Q'\subsetneq Q}}
\frac{v(G\cup Q' \bs P' )}{v(G)}\frac{v(G\cup \bar{Q}' \bs \bar{P}' )}{v(G)}~|P'|~(-1)^{|\bar{P}'|+1}~\sigma(\bar{P}', P') ~\sigma(Q',\bar{Q}').
\label{eq:gen_wick_iter_tilde}
\end{align}

\end{proof}

\section {Generalized Wick's rule}
\label{app:cumulant_exp}

Consider a state $\ket{v}=\sum v(S)$ that lies in $\GG_k\subset \mathcal{H}$ and three ordered sequences $G=(1,..n)$, $P\subset G$, $Q\subset [l]\bs{G}$. We are concerned with decomposing multi-excitation amplitudes $v(G\cup Q \bs P )/v(G)$ into few-excitation $\frac{v(G\cup Q' \bs P' )}{v(G)}$ ($|P'|=|Q'|\leq k$). As was mentioned in the main text, giving an explicit formula for such a decomposition is difficult. The key simplifying step is to use `connected amplitudes' $v^{(c)}_{P,Q}$:
\begin{align}
v^{(c)}_{P,Q}\defeq
\frac{v(G\cup Q \bs P )}{v(G)}
-
\sum_{\substack{P'\subsetneq P,~ Q'\subsetneq Q}}
\frac{v(G\cup Q' \bs P' )}{v(G)}\frac{v(G\cup \bar{Q}' \bs \bar{P}' )}{v(G)}~
(-1)^{|\bar{P}'|+1}\,\frac{|P'|}{|P|}\,\sigma(\bar{P}', {P}')~\sigma(Q', \bar{Q}').\label{eq:v_c_def}
\end{align}
which capture the deviation of $\ket{v}$ from recursive Wick's rule (Eq.\,\ref{eq:gen_wick_iter}) --- if $\ket{v}\in\GG_k$ and $|P|=|Q|>k$, $v^{(c)}_{P,Q}=0$. For $|P|=|Q|=1$, we define simply $v^{(c)}_{P,Q}=\frac{v(G\cup Q \bs P )}{v(G)}$. Theorem\,\ref{thm:thm_corr_decomposition} below will expresses amplitudes $\frac{v(G\cup Q \bs P )}{v(G)}$ of $\ket{v}\in \GG_k$ in terms of $v^{(c)}_{P',Q'}$ for $|P'|=|Q'|\leq k$ alone. One can think of this statement as an amplitude version of a cumulant expansion. Via Eq.\,\ref{eq:v_c_def}, this implies a decomposition of $\frac{v(G\cup Q \bs P )}{v(G)}$ into $\frac{v(G\cup Q' \bs P' )}{v(G)}$ for $|P'|=|Q'|<k$. 

To formally state the theorem, further terminology needs to be introduced. We denote $\mathrm{Part}(P,Q)$ the set of `partitions of $(P,Q)$'. Namely, $\mathrm{Part}(P,Q)$ consists of all sets $\mathcal{R}$ of type $\mathcal{R}=\{(P'_1,Q'_1),..(P'_{|\mathcal{R}|},Q'_{|\mathcal{R}|})\}$ such that $|P'_a|=|Q'_a|$, and for (disjoint!)  unions $P(\mathcal{R})\equiv \bigcup^{|\mathcal{R}|}_{a=1} P'_a=P$ and $Q(\mathcal{R})\equiv\bigcup^{|\mathcal{R}|}_{a=1} Q'_a=Q$. A more refined set $\mathrm{Part}_k(P,Q)\subset \mathrm{Part}(P,Q)$ is given by applying another constraint $|P'_a|,|Q'_a|\leq k$ to all $(P'_a,Q'_a)\in \mathcal {R}\in \mathrm{Part}(P,Q)$. The signature function for $\mathcal{R}$ is defined as 
\begin{equation}
\label{eq:sigma_R_def}
\sigma\left(\mathcal{R}\right)=\sigma\left(P'_{|\mathcal{R}|},P'_{|\mathcal{R}|-1},..,P'_1\right) \,\sigma\left(Q'_1,Q'_2,..,Q'_{|\mathcal{R}|}\right).
\end{equation}
Next, we introduce a vector $\mathbf{m}(\mathcal{R})=(m_1,m_2,..,m_k)$, where each $m_{k'}\in \mathbb{N}_+$ gives the number of tuples $(P',Q')$ in $\mathcal{R}$ such that $|P|=|Q|=k'$. In other words, $\mathbf{m}$ encodes an integer partition of $|P|=|Q|$ that is defined by $\mathcal{R}$ (which itself is a `partition of $(P,Q)$'). For a general vector $\mathbf{m}$, we denote one-norm $\sum_{k'} m_{k'}$ as $|\mathbf{m}|$, and also will make use of the expression $\mathbf{k}\cdot\mathbf{m}\defeq \sum_{k'} k' m_{k'}$. If $\mathbf{m}=\mathbf{m}(\mathcal{R})$ for $\mathcal{R}\in \mathrm{Part}_k(P,Q)$, we have simply $|\mathbf{m}|=|\mathcal{R}|$ and $\mathbf{k}\cdot\mathbf{m}=|P|$. We say $\mathbf{m}'<\mathbf{m}$ if for all $k'$ holds $m'_{k'}\leq m_{k'}$ and at least for one $k'$ one has strictly $m'_{k'}< m_{k'}$.

\begin{theorem}
\label{thm:thm_corr_decomposition}
For any state $\ket{v}\in \GG_k$, the decomposition holds:
\begin{equation}
\frac{v(G\cup Q \bs P )}{v(G)}=\sum_{\substack{\mathcal{R}\in\mathrm{Part}_k(P,Q)}}\nu(\mathbf{m}(\mathcal{R}))~\sigma(\mathcal{R})\prod_{(P',Q')\in\mathcal{R}} v^{(c)}_{{P'}, {Q'}}.
\label{eq:thm_corr_decomposition}
\end{equation}
Here the function $\nu(\mathbf{m})$ is equal to 1 if $|\mathbf{m}|=0$ or $|\mathbf{m}|=1$, and otherwise is defined by the recursive relation
\begin{align}
\nu(\mathbf{m})=\sum_{\mathbf{m}'<\mathbf{m}}   (-1)^{\mathbf{k}\cdot(\mathbf{m}-\mathbf{m}')+1} \frac{\mathbf{k}\cdot \mathbf{m}'}{\mathbf{k}\cdot \mathbf{m}} ~\nu(\mathbf{m}')~\nu(\mathbf{m}-\mathbf{m}')~\begin{pmatrix}m_1' \\ m_1\end{pmatrix}
\cdot ..\cdot 
\begin{pmatrix}m_k' \\ m_k \end{pmatrix}.
\label{eq:nu_def}
\end{align}
\end{theorem}
\begin{proof}
We will show Eq.\,\ref{eq:thm_corr_decomposition} for $\ket{v}\in \GG_k$ by first proving a general statement for all $\ket{v}\in\mathcal{H}$:
\begin{equation}
\frac{v(G\cup Q \bs P )}{v(G)}=\sum_{\substack{\mathcal{R}\in\mathrm{Part}(P,Q)}}\nu(\mathbf{m}(\mathcal{R}))~\sigma(\mathcal{R})\prod_{(P',Q')\in\mathcal{R}} v^{(c)}_{{P'}, {Q'}}.
\label{eq:corr_decomposition_gen}
\end{equation}
Note that coefficients $v^{(c)}_{{P'}, {Q'}}$ are zero for $|P'|,|Q'|>k$ if $\ket{v}\in \GG_k$. Since the contributions of $\mathcal{R}\in \mathrm{Part}(P,Q)\bs \mathrm{Part}_k(P,Q)$ to Eq.\,\ref{eq:corr_decomposition_gen} are proportional to at least one such coefficient, these contributions vanish; therefore, Eq.\,\ref{eq:corr_decomposition_gen} yields Eq.\,\ref{eq:thm_corr_decomposition}.

We now prove Eq.\,\ref{eq:corr_decomposition_gen} by induction in $|P|$. For the base of induction $|P|=|Q|=1$, by definition we have $\frac{v(G\cup Q \bs P )}{v(G)}=v^{(c)}_{{P}, {Q}}$. Since $\rm{Part}(P,Q)$ in this case only consists of the trivial $\mathcal{R}=\{(P,Q)\}$, for which $\sigma(\mathcal{R})=1$ and $\nu(\mathbf{m}(\mathcal{R}))=1$, we recover Eq.\,\ref{eq:corr_decomposition_gen} directly. We now assume the validity of Eq.\,\ref{eq:corr_decomposition_gen} for $|P|=|Q|=k_{\mathrm{ind}}$ and prove it for $|P|=|Q|=k_{\mathrm{ind}}+1$. By definition in Eq.\,\ref{eq:v_c_def} and induction step, we have (denoting $\bar{Q}'=Q\bs Q'$ and $\bar{P}'=P\bs P'$):
\begin{align}\notag
\frac{v(G\cup Q \bs P )}{v(G)}=v^{(c)}_{{P}, {Q}}~+
&
\sum_{\substack{P'\subsetneq P,~ Q'\subsetneq Q\\
|P'|=|Q'|}} \sum_{\substack{\mathcal{R}'\in\mathrm{Part}(P',Q')\\
\bar{\mathcal{R}}'\in\mathrm{Part}(\bar{P}',\bar{Q}')}}
\left(\prod_{(P'',Q'')\in\mathcal{R}'} 
v^{(c)}_{{P'}, {Q'}}\prod_{(\bar{P}'',\bar{Q}'')\in\bar{\mathcal{R}}'} v^{(c)}_{{\bar{P}'}, {\bar{Q}'}}\right)\\&
\times (-1)^{|\bar{P}'|+1}\,\frac{|P'|}{|P|}~{\nu(\mathbf{m}(\mathcal{R}'))}{\nu(\mathbf{m}(\bar{\mathcal{R}}'))}~\sigma(\mathcal{R}')\sigma(\bar{\mathcal{R}}')\sigma(\bar{P}',P')\sigma(\bar{Q}',Q')
\label{eq:R_induction}
\end{align}

Using the property of the sign function $\sigma(A)\sigma(B)\sigma(A, B)=\sigma(A\cup B)$ and the definition of $\sigma(\mathcal{R})$ (Eq.\,\ref{eq:sigma_R_def}), one finds a simplification $\sigma(\mathcal{R}')\sigma(\bar{\mathcal{R}}')\sigma(\bar{P}',P')\sigma(Q', \bar{Q'})=\sigma(\mathcal{R}'\cup \bar{\mathcal{R}}')$. 

Let us show that expression Eq.\,\ref{eq:R_induction} reproduces the desired sum in Eq.\,\ref{eq:corr_decomposition_gen}. The term for $\mathcal{R}=\{(P,Q)\}$ of this sum is given directly by $v^{(c)}_{{P}, {Q}}$ in Eq.\,\ref{eq:R_induction}, since for such $\mathcal{R}$ the coefficient ${\nu(\mathbf{m}(\mathcal{R}))}~\sigma(\mathcal{R})$ is equal to $1$. To reproduce the rest of the sum in Eq.\,\ref{eq:corr_decomposition_gen}, consider the products $\left(\sigma(\mathcal{R}'\cup \bar{\mathcal{R}}')\prod_{(P'',Q'')\in\mathcal{R}'} 
v^{(c)}_{{P''}, {Q''}}\prod_{(\bar{P}'',\bar{Q}'')\in\bar{\mathcal{R}}''} v^{(c)}_{{\bar{P}''}, {\bar{Q}''}}\right)$ in Eq.\,\ref{eq:R_induction}. These have the form $\sigma(\mathcal{R})\prod_{({P}'',{Q}'')\in\mathcal{R}} v^{(c)}_{{{P}''}, {{Q}''}}$ for $\mathcal{R}\equiv\mathcal{R}'\cup \bar{\mathcal{R}}'$. Such $\mathcal{R}$ is a nontrivial ($\mathcal{R}\neq\{(P,Q)\}$) partition of $(P,Q)$. Let us now determine the coefficient in front of such a product for any $\mathcal{R}$, examining the terms in Eq.\,\ref{eq:R_induction} coming from all pairs $\mathcal{R}',~\bar{\mathcal{R}}'$ which yield $\mathcal{R}'\cup\bar{\mathcal{R}}'=\mathcal{R}$. We observe that for any nontrivial $\mathcal{R}$ every possible splitting into $\mathcal{R}',~\bar{\mathcal{R}}'$ appears in the sum in Eq.\,\ref{eq:R_induction} exactly once --- in the component of the sum where $P'=P(\mathcal{R}')$, $Q'=Q(\mathcal{R}')$. Collecting the factor in front of $\sigma(\mathcal{R}) \prod_{(P'',Q'')\in\mathcal{R}}  v^{(c)}_{{{P}''}, {{Q}''}}$ in Eq.\,\ref{eq:R_induction}, we find
\begin{align}\notag
&\sum_{\mathcal{R}'\subsetneq\mathcal{R}}(-1)^{|P(\mathcal{R}')|-|P(\mathcal{R})|+1}\dfrac{|P(\mathcal{R}')|}{|P(\mathcal{R})|} {\nu(\mathbf{m}(\mathcal{R}'))}{\nu(\mathbf{m}(\mathcal{R}\bs {\mathcal{R}}'))}\\
\notag
=&\sum_{\mathbf{m}'<\mathbf{m}(\mathcal{R})}(-1)^{\mathbf{k}\cdot(\mathbf{m}(\mathcal{R})-\mathbf{m}')+1}\dfrac{\mathbf{k}\cdot\mathbf{m}'}{\mathbf{k}\cdot\mathbf{m}(\mathcal{R})} {\nu(\mathbf{m}')}{\nu(\mathbf{m}(\mathcal{R})-\mathbf{m}')}\cdot 
\begin{pmatrix}m_1' \\ m_1(\mathcal{R})\end{pmatrix}
\cdot ..\cdot 
\begin{pmatrix}m_k' \\ m_k(\mathcal{R})\end{pmatrix}\\
=& {\nu(\mathbf{m}(\mathcal{R}))}
\end{align}
Thus restoring the coefficient in Eq.\,\ref{eq:corr_decomposition_gen}, we show that the sum in Eq.\,\ref{eq:R_induction} gives the part of the sum over $\mathcal{R}$ in Eq.\,\ref{eq:corr_decomposition_gen} over nontrivial $\mathcal{R}\neq \{(P,Q)\}$. 

\end{proof}

\section{Non-Slater ansatz derivation}
\label{app:ansatz}

In this section we prove that Theorem\,\ref{thm:thm_corr_decomposition} directly implies an ansatz for states $\ket{v}\in \GG_k$:

\begin{theorem}
Any state $\ket{v}\in \GG_k$ can be represented as:
\begin{align}
\label{eq:hyperferm_ansatz} \ket{v}&=v(G)\, F(T_1,..T_k)\ket{G},\\
\label{eq:CC_gens} T_{k'}&=\sum_{\substack{P\subset G, Q\subset [l]\bs{G},\\ |P|=|Q|=k'}}  \theta_{P,Q} \Psi^\dag_{Q} \Psi_{P}.
\end{align}
for parameters $\theta_{P,Q}$ chosen as $v^{(c)}_{P,Q} \, \sigma(G\bs P, P)$ ($v^{(c)}_{P,Q}$ defined in Appendix\,\ref{app:cumulant_exp}). The function $F(x_1,..x_k)$ is defined as
\begin{equation}
    \label{eq:F_func_form}
    F(x_1, x_2,...) = \sqrt{1+ 2(x_2+x_4+..)}~\exp \left( \int\limits_0^{1} \dfrac{x_1+3x_3\mu^2+5x_5 \mu^4+..}{1+ 2(x_2 \mu^2 +x_4 \mu^4 +..)} d\mu \right).
\end{equation}

\label{thm:ansatz}
\end{theorem}

\begin{proof}
The proof consists of two parts. First we show from Theorem\,\ref{thm:thm_corr_decomposition} that for $\ket{v}\in \GG_k$ one has
\begin{align}
\ket{v}=v(G)\sum_{\mathbf{m}} \frac{\nu(\mathbf{m})}{m_1!m_2!..m_k!} T_1^{m_1}T_2^{m_2}..T_k^{m_k}\ket{G}
\label{eq:hyperferm_ansatz_Taylor}
\end{align}
for coefficient $\nu(\mathbf{m})$ defined in Theorem\,\ref{thm:thm_corr_decomposition} and parameters in $T_{k'}$ chosen as $\theta_{P,Q}={v}^{(c)}_{P,Q} \sigma(G\bs P, P)$ for $v^{(c)}_{P,Q}$ defined in Appendix\,\ref{app:cumulant_exp}. Subsequently, we show that the generating function of $\nu(\mathbf{m})$, defined as
\begin{align}
F(x_1, .., x_k) = \sum_{\mathbf{m}} \frac{\nu(\mathbf{m})}{m_1!m_2!..m_k!} x_1^{m_1}..x_k^{m_k},
\label{eq:F_Taylor_def}
\end{align}
has the functional form of Eq.\,\ref{eq:F_func_form}. This will conclude the proof, showing that Eq.\,\ref{eq:hyperferm_ansatz_Taylor} reproduces Eq.\,\ref{eq:hyperferm_ansatz} from Theorem\,\ref{thm:ansatz}. Here we use the fact that $T_{k'}$ operators mutually commute, as all operators of form $\Psi^\dag_{Q} \Psi_{P}$ for $|Q|=|P|$ and $P\subset G$, $Q\subset [l]\bs G$ mutually commute. 

To the reader comparing Theorems \,\ref{thm:thm_corr_decomposition} and \ref{thm:ansatz}, the formulas in Eqs.\,\ref{eq:F_func_form} and \ref{eq:F_Taylor_def} might come as a surprise. There is an apparent contradiction: the function $F(x_1,..x_k)$ has an infinite Taylor series, while the sum over $\mathbf{m}$ in Theorem\,\ref{thm:thm_corr_decomposition} is finite. Indeed, $\mathbf{m}$ come from partitions of potentially large but bounded-size excitations. This confusion can be clarified by observing that the operators $T_k$ are nilpotent; this is because operators $\Psi^\dag_{Q} \Psi_{P}$ for $P\subset G$, $Q\subset [l]\setminus G$ square to zero.

Let us now show that Eq.\,\ref{eq:hyperferm_ansatz_Taylor} is indeed equivalent to the established Theorem\,\ref{thm:thm_corr_decomposition}. For this, we expand $T_{k'}$ in Eq.\,\ref{eq:hyperferm_ansatz_Taylor} noting the mutual commutation of $\Psi^\dag_{Q} \Psi_{P}$ operators. This yields
\begin{align}
\ket{v}=v(G)\sum_{\mathcal{R}} \nu(\mathbf{m}(\mathcal{R})) \left(\prod_{(P,Q)\in \mathcal{R}}\theta_{P,Q}\Psi^\dag_Q\Psi^{\phantom{\dagger}}_P\right)\ket{G}
\label{eq:hyperferm_ansatz_Taylor_2}
\end{align}
where the sum runs over all sets of type $\mathcal{R}=\{(P_1,Q_1), .. , (P_{|\mathcal{R}|},Q_{|\mathcal{R}|})\}$ for $|P_a|=|Q_a|\leq k$; function $\mathbf{m}(\mathcal{R})$ is defined as in Appendix\,\ref{app:cumulant_exp}. Factorials in the denominator are cancelled due to the term corresponding to $\mathcal{R}$ represented $m_1!..m_k!$ times.
We further observe: 
\begin{align}
\left(\prod_{(P,Q)\in \mathcal{R}}\Psi^\dag_Q\Psi^{\phantom{\dagger}}_P\right)&=
\Psi^\dag_{Q_1}\Psi^{\phantom{\dagger}}_{P_1}\Psi^\dag_{Q_2}\Psi^{\phantom{\dagger}}_{P_2}..\Psi^\dag_{Q_{|\mathcal{R}|}}\Psi^{\phantom{\dagger}}_{P_{|\mathcal{R}|}}\notag\\
&=\Psi^\dag_{Q_1}\Psi^\dag_{Q_2}..\Psi^\dag_{Q_{|\mathcal{R}|}}\Psi^{\phantom{\dagger}}_{P_{|\mathcal{R}|}}..\Psi^{\phantom{\dagger}}_{P_2}\Psi^{\phantom{\dagger}}_{P_1}\notag\\
&=\Psi^\dag_{\cup_a Q_a}\Psi^{\phantom{\dagger}}_{\cup_a P_a}~\sigma(Q_{|\mathcal{R}|}, ..Q_2,Q_1) ~\sigma(P_{|\mathcal{R}|}, ..P_2,P_1)
\end{align}
Substituting into Eq.\,\ref{eq:hyperferm_ansatz_Taylor_2} and observing $\Psi^\dag_{Q}\Psi^{\phantom{\dagger}}_{P}\ket{G}=\sigma(G\bs P, P) \ket{G\cup Q \bs P}$ yields
\begin{align}
\notag
\ket{v}=v(G)\sum_{\mathcal{R}} &~\nu(\mathbf{m}(\mathcal{R}))~\sigma(G\bs P(\mathcal{R}), P(\mathcal{R})) ~\sigma(Q_{|\mathcal{R}|}, ..Q_2,Q_1)~\sigma(P_{|\mathcal{R}|}, ..P_2,P_1)\\
&\times \left(\prod_{(P',Q')\in \mathcal{R}}\theta_{P',Q'}\right)\ket{G\cup Q(\mathcal{R}) \bs P(\mathcal{R}) }
\label{eq:hyperferm_ansatz_Taylor_3}
\end{align}
for $P(\mathcal{R})=\cup_a P_a,~Q(\mathcal{R})=\cup_a Q_a$ (as in Appendix\,\ref{app:cumulant_exp}). We now claim
\begin{align}
    \sigma(G\bs P(\mathcal{R}), P(\mathcal{R})) =\sigma(P_{|\mathcal{R}|},..P_{2}, P_1)~\sigma (P_1, P_2,..P_{|\mathcal{R}|}) ~ \prod_a \sigma(G\bs P_a, P_a)
    \label{eq:G_P_chain}
\end{align}
by induction. Indeed, for $|\mathcal{R}|=1$ this holds trivially. Induction step from $(|\mathcal{R}|-1)$ to $|\mathcal{R}|$ reads (we use properties from Eqs.\,\ref{eq:sigma_rule_1}-\ref{eq:sigma_rule_4} for transformations)
\begin{align}
    \notag \cup^{|\mathcal{R}|}_{a=1} \sigma(G\bs P_a, P_a)=&\sigma(G\bs P_{|\mathcal{R}|},P_{|\mathcal{R}|}) \sigma(G\bs (\cup^{|\mathcal{R}|-1}_{a=1} P_a), \cup^{|\mathcal{R}|-1}_{a=1} P_a) ~\\
    \notag & \times \sigma(P_{|\mathcal{R}|-1},..P_{2}, P_1)~\sigma (P_1, P_2,..P_{|\mathcal{R}|-1})\\
    \notag =&\sigma(G\bs P_{|\mathcal{R}|},P_{|\mathcal{R}|})~\sigma(G\bs (\cup^{|\mathcal{R}|}_{a=1} P_a), \cup^{|\mathcal{R}|-1}_{a=1} P_a)~\sigma(P_{|\mathcal{R}|}, \cup^{|\mathcal{R}|-1}_{a=1} P_a)\\
    \notag & \times \sigma(P_{|\mathcal{R}|-1},..P_{2}, P_1)~\sigma (P_1, P_2,..P_{|\mathcal{R}|-1})\\
    \notag =&\sigma(G\bs (\cup^{|\mathcal{R}|}_{a=1} P_a), \cup^{|\mathcal{R}|}_{a=1} P_a)~\sigma(\cup^{|\mathcal{R}|-1}_{a=1} P_a, P_{|\mathcal{R}|})\sigma(P_{|\mathcal{R}|}, \cup^{|\mathcal{R}|-1}_{a=1} P_a)\\
    \notag & \times \sigma(P_{|\mathcal{R}|-1},..P_{2}, P_1)~\sigma (P_1, P_2,..P_{|\mathcal{R}|-1})\\
    \notag =&\sigma(G\bs (\cup^{|\mathcal{R}|}_{a=1} P_a), \cup^{|\mathcal{R}|}_{a=1} P_a)~\sigma(P_{|\mathcal{R}|},..P_{2}, P_1)~\sigma (P_1, P_2,..P_{|\mathcal{R}|})
\end{align}
Employing Eq.\,\ref{eq:G_P_chain} and substituting $\theta_{P',Q'}=v^{(c)}_{P',Q'} \sigma(G\bs P', P')$ in Eq.\,\ref{eq:hyperferm_ansatz_Taylor_3}, we obtain
\begin{align}
\notag \ket{v}=v(G)\sum_{\mathcal{R}} &~\nu(\mathbf{m}(\mathcal{R}))~\sigma(Q_{|\mathcal{R}|}, ..,Q_2,Q_1)~\sigma(P_1,P_2,..,P_{|\mathcal{R}|})\\
&\times \left(\prod_{(P',Q')\in \mathcal{R}}v^{(c)}_{P',Q'}\right)\ket{G\cup Q(\mathcal{R}) \bs P(\mathcal{R}) }
\label{eq:hyperferm_ansatz_Taylor_4}
\end{align}
Grouping terms proportional to identical basis states $\ket{G\cup Q(\mathcal{R}) \bs P(\mathcal{R}) }$ yields
\begin{align}
\ket{v}=v(G)\sum_{P\subset G,Q\subset [l]\bs G}\ket{G\cup Q \bs P}\sum_{\mathcal{R}\in \mathrm{Part}_k(P,Q)} &~\nu(\mathbf{m}(\mathcal{R}))~\sigma(\mathcal{R})\left(\prod_{(P',Q')\in \mathcal{R}}v^{(c)}_{P',Q'}\right)
\label{eq:hyperferm_ansatz_Taylor_final}
\end{align}
with the set of partitions $\mathrm{Part}_k(P,Q)$ and sign $\sigma(\mathcal{R})$ defined in Appendix\,\ref{app:cumulant_exp}. Comparing Eq.\,\ref{eq:hyperferm_ansatz_Taylor_final} to Theorem\,\ref{thm:thm_corr_decomposition}, we observe that these statements are identical. As the transformations from Eq.\,\ref{eq:hyperferm_ansatz_Taylor} to Eq.\,\ref{eq:hyperferm_ansatz_Taylor_final} were all equivalences, Theorem\,\ref{thm:thm_corr_decomposition} implies validity of Eq.\,\ref{eq:hyperferm_ansatz_Taylor}. 

We now show that function $F(x_1,..,x_k)$ from Eq.\,\ref{eq:F_Taylor_def} satisfies the definition from Eq.\,\ref{eq:F_func_form}. By definition of $\nu(\mathbf{m})$ (see Theorem\,\ref{thm:thm_corr_decomposition}), $\nu(\mathbf{m})=1$ if $|\mathbf{m}|=1$ or $|\mathbf{m}|=0$ and for $|\mathbf{m}|>1$ is defined by the relation
\begin{align}
(\mathbf{k}\cdot \mathbf{m})~ \nu(\mathbf{m})=\sum_{\mathbf{m}'<\mathbf{m}}   (-1)^{\mathbf{k}\cdot(\mathbf{m}-\mathbf{m}')+1} {(\mathbf{k}\cdot \mathbf{m}')}{} ~\nu(\mathbf{m}')~\nu(\mathbf{m}-\mathbf{m}')~\begin{pmatrix}m_1' \\ m_1\end{pmatrix}
\cdot ..\cdot 
\begin{pmatrix}m_k' \\ m_k \end{pmatrix}.
\label{eq:nu_def_copy}
\end{align}
Multiplying both sides by $\frac{x_1^{m_1}..x_k^{m_k}}{m_1!..m_k!}$ and summing up the equations for all $|\mathbf{m}|>1$, we obtain
\begin{align}
\sum_{\mathbf{m}, |\mathbf{m}|>1} \frac{(\mathbf{k}\cdot \mathbf{m})~\nu(\mathbf{m}) }{m_1!..m_k!} x_1^{m_1}..x_k^{m_k}=\sum_{\mathbf{m}, |\mathbf{m}|>1}\sum_{\mathbf{m}'<\mathbf{m}}   (-1)^{\mathbf{k}\cdot(\mathbf{m}-\mathbf{m}')+1} {\mathbf{k}\cdot \mathbf{m}'}{} ~\frac{\nu(\mathbf{m}')}{m_1!..m_k!} ~\frac{\nu(\mathbf{m}-\mathbf{m}')~x_1^{m_1}..x_k^{m_k}}{(m_1-m'_1)!..(m_k-m'_k)!}.
\label{eq:gen_func_eq_1}
\end{align}
Changing the summation variables on the right hand side from $\mathbf{m}$ to $\mathbf{m}''=\mathbf{m}-\mathbf{m}'$ yields
\begin{align}
\sum_{\mathbf{m}, |\mathbf{m}|>1} \frac{(\mathbf{k}\cdot \mathbf{m})~\nu(\mathbf{m})}{m_1!..m_k!} x_1^{m_1}..x_k^{m_k}&=\sum_{\mathbf{m}'}  {}~\frac{(\mathbf{k}\cdot \mathbf{m}')~\nu(\mathbf{m}')}{m_1!..m_k!}x_1^{m_1'}..x_k^{m_k'}~\sum_{\mathbf{m}'', |\mathbf{m}''|>0}\frac{(-1)^{\mathbf{k}\cdot \mathbf{m}''+1}\nu(\mathbf{m}'')~}{m''_1!..m''_k!}x_1^{m_1''}..x_k^{m_k''},\notag\\
-(x_1+2x_2+..+kx_k)&=\sum_{\mathbf{m}'}  {}~\frac{(\mathbf{k}\cdot \mathbf{m}')~\nu(\mathbf{m}')}{m_1!..m_k!}x_1^{m_1'}..x_k^{m_k'}~\left(-1+\sum_{\mathbf{m}'', |\mathbf{m}''|>0}\frac{(-1)^{\mathbf{k}\cdot \mathbf{m}''+1}\nu(\mathbf{m}'')~}{m''_1!..m''_k!}x_1^{m_1''}..x_k^{m_k''}\right).\notag
\end{align}
In the second line we used the fact that $\nu(\mathbf{m})=1$ for $|\mathbf{m}|=1$. Flipping the sign on both sides and using $F(x_1,..,x_k)$ from Eq.\,\ref{eq:F_Taylor_def}, we obtain

\begin{align}
    \sum^k_{k'=1} k' x_{k'}=\left(\sum^k_{k'=1} k' x_{k'} \frac{\partial}{\partial x_{k'}}F(x_1,..x_k)\right) {F(-x_1,x_2,-x_3, .. (-1)^k x_k)} \label{eq:gen_func_eq_2}
\end{align}
We now solve for $F(x_1,..x_k)$ with initial condition $F(0,0,...) = 1$ (as $\nu (\mathbf{m})=1$ for $\mathbf{m}=\mathbf{0}$). Rewrite Eq.\,\ref{eq:gen_func_eq_2} as
\begin{equation}
\label{eq:generfunclambdaeq}
f(-\lambda) \dfrac{df(\lambda)}{d\lambda} = p'(\lambda)    
\end{equation}
where
\begin{equation}
    f(\lambda) = F(x_1\lambda, x_2\lambda^2,x_3\lambda^3,...), ~~~ p(\lambda) = \sum_{k'=1}^{k} \lambda^{k'} x_{k'}
\end{equation}
with the rescaling $x_{k'} \mapsto x_{k'} \lambda^{-{k'}}$. Decomposing $p$ into symmetric and anti-symmetric parts
\begin{equation}
    a(\lambda)  = x_{1} \lambda + x_{3} \lambda^{3} + ..., ~~~ s(\lambda) = \frac{1}{2} + x_2 \lambda^2 + x_4 \lambda^4 + ...
\end{equation}
and combining Eq.\,\ref{eq:generfunclambdaeq} with the one with $\lambda \mapsto -\lambda$, one obtains
\begin{equation}
    \dfrac{d}{d\lambda} (f(\lambda) f(-\lambda)) = 2 s'(\lambda)
\end{equation}
This gives $f(\lambda) f(-\lambda) = 2s(\lambda)$, which satisfies initial conditions at $\lambda = 0$. Now one can solve Eq.\,\ref{eq:generfunclambdaeq} by
\begin{equation}
    \dfrac{2 s}{f} \dfrac{df}{d\lambda} = (s' + a') ~~~ \Rightarrow ~~~ f(\lambda) = \exp \left( \int \dfrac{s' + a'}{2s} d\lambda \right) = \sqrt{2s(\lambda)}\exp \left( \int\limits_0^{\lambda} \dfrac{a'(\mu)}{2s(\mu)} d\mu \right)
\end{equation}
so the generating function is
\begin{equation}
    F(x_1, x_2,...) = \sqrt{2s(1)}\exp \left( \int\limits_0^{1} \dfrac{a'(\mu)}{2s(\mu)} d\mu \right).
\end{equation}
This coincides with Eq.\,\ref{eq:F_func_form}, as desired.

\end{proof}


\begin{thebibliography}{99}
\bibitem{ftnt_oddeven} Using a symbolic math package, we find that components of Eq.\,\ref{eq:k_Plücker} for $k=2r-1$ and $2r$ are equivalent for all sizes tested. Proving this analytically, however, is challenging.
\bibitem{ftnt_cup_bs} In $G\cup Q \bs P $ and other expressions without brackets, the earlier operations (here, $\cup$) are to be applied first.
\end{thebibliography}
\end{document}